\documentclass[11pt, draftclsnofoot, onecolumn]{IEEEtran}
\usepackage[T1]{fontenc}
\usepackage[latin9]{inputenc}
\usepackage{geometry}
\usepackage{float}
\usepackage{enumitem}
\usepackage{amsmath}
\usepackage{amsthm}
\usepackage{amssymb}
\usepackage{setspace}
\usepackage{xcolor}
\usepackage{cite}
\usepackage{graphicx}
  \usepackage[caption=false,font=footnotesize,labelformat=empty]{subfig}
\makeatletter

\floatstyle{ruled}
\newfloat{algorithm}{tbp}{loa}
\providecommand{\algorithmname}{Algorithm}
\floatname{algorithm}{\protect\algorithmname}

\theoremstyle{plain}
\newtheorem{thm}{\protect\theoremname}
\theoremstyle{plain}
\newtheorem{prop}[thm]{\protect\propositionname}
\theoremstyle{remark}
\newtheorem{rem}[thm]{\protect\remarkname}
\theoremstyle{plain}
\newtheorem{lem}[thm]{\protect\lemmaname}

\usepackage{hyperref}
\usepackage{enumitem}
\usepackage{breqn}
\usepackage{bbm} 
\usepackage{cite}
\usepackage{algorithm,algpseudocode}

\DeclareMathOperator*{\argmax}{arg\,max}

\allowdisplaybreaks

\global\long\def\s[#1]{\textnormal{\scriptsize #1}}
\global\long\def\st[#1]{\textnormal{\tiny #1}}

\global\long\def\Dkl{\mathrm{D_{KL}}}
\global\long\def\ucb{\mathrm{UCD}}

\global\long\def\P{\mathbb{P}}
\global\long\def\E{\mathbb{E}}

\global\long\def\I{\mathbbm{1}}
\global\long\def\v[#1]{\mathbf{#1}} % vectors
\global\long\def\m[#1]{\boldsymbol{#1}} % matrices, or collection of vectors
\global\long\def\dtv{\mathrm{d_{TV}}}

\global\long\def\r[#1]{#1}
\global\long\def\d{\mathrm{d}}

\global\long\def\dfn{:=}

\global\long\def\trre[#1,#2]{\overset{{\scriptstyle (#2)}}{#1}} % transition explained with reason

\author{Nir Weinberger$^{1}$ and Michal Yemini$^{2}$
\thanks{
$^1$ Nir Weinberger is with the Faculty of 
Electrical and Computer Engineering, Technion - Israel Institute of Technology. (Email: nirwein@technion.ac.il.)\\
$^2$Michal Yemini is  with the Faculty of 
Electrical and Computer Engineering, Princeton University. (Email: myemini@princeton.edu.)\\
This paper was presented in part at the  2022 IEEE International Symposium on Information Theory (ISIT), \cite{ISIT2022_WY}.}
}

\makeatother

\usepackage{babel}
\providecommand{\lemmaname}{Lemma}
\providecommand{\corollaryname}{Corollary}
\providecommand{\propositionname}{Proposition}
\providecommand{\remarkname}{Remark}
\providecommand{\theoremname}{Theorem}

\begin{document}

% \title{Upper Confidence Interval Strategies for Multi-Armed Bandits with Entropy Rewards}

\title{Multi-Armed Bandits with\\ Self-Information  Rewards}
% We can also try something like 

\maketitle
\renewcommand\[{\begin{equation}}
\renewcommand\]{\end{equation}}

\begin{abstract}

   This paper introduces the informational multi-armed bandit (IMAB) model in which at each round, a player chooses an arm, observes a symbol, and receives an unobserved reward in the form of the symbol's self-information. Thus, the expected reward of an arm is the Shannon entropy of the probability mass function of the source that generates its symbols. The player aims to maximize the expected total reward associated with the entropy values of the arms played. 
   Under the assumption that the alphabet size is known, two UCB-based algorithms are proposed for the IMAB model which consider the biases of the plug-in entropy estimator. The first algorithm optimistically corrects the bias term in the entropy estimation. The second algorithm relies on data-dependent confidence intervals that adapt to sources with small entropy values. Performance guarantees are provided by upper bounding the expected regret of each of the algorithms. Furthermore, in the Bernoulli case, the  asymptotic behavior of these algorithms is compared to the Lai-Robbins lower bound for the pseudo regret. Additionally, under the assumption that the \textit{exact} alphabet size is unknown, and instead the player only knows a loose upper bound on it,  a UCB-based algorithm is proposed, in which  the player aims to reduce the regret caused by the unknown alphabet size in a finite time regime.  Numerical results illustrating the expected regret of the algorithms presented in the paper are provided.

\end{abstract}


\section{Introduction \label{sec:Introduction}}

Multi-armed bandit (MAB) problems are sequential decision problems where a player makes iterative decisions in an unfamiliar environment to optimize a total outcome. More specifically, at every round the player is given a choice of $K$ arms, each affiliated with an unknown probability mass function (PMF) for its reward. The player chooses an arm to play based on its previous arm choices and received rewards, and then receives a random reward generated by the chosen arm. The player's objective is to maximize the total expected reward it receives from all the rounds it has played.
If the player knew the expected reward of each arm, it could maximize its total expected reward by repeatedly choosing the arm with the highest expected reward. However, since the player does not know in advance the expected reward of each arm, it must balance two conflicting acts, namely, \textit{exploration and exploitation}. 
When making an arm choice, the player wants to exploit the knowledge it accumulated and choose the arm with the highest expected reward, however, naively choosing repeatedly the arm with the highest estimated reward can be sub-optimal since this estimate can be erroneous.
To that end, the player periodically dedicates rounds to exploration, aiming to increase the estimation precision of the expected rewards. Balancing the wish to exploit current observations and maximize the immediate reward with the need to explore other arms to increase estimation precision and thus future rewards lies at the heart of MAB decision algorithms; it is known as the exploration-exploitation trade-off. 

In the classical MAB problem \cite{lai1985asymptotically}, the reward of an arm is independently and identically distributed over different rounds, and so the expected reward of each arm is the mean of its reward distribution. Furthermore, it can be estimated by the sample mean of the observed rewards which is an unbiased estimator. 
The classical model has been extended in numerous ways to include, among other models, MAB with linear reward functions \cite{doi:10.1287/moor.1100.0446,7472588,Dani2008StochasticLO,NIPS2011_4417}, Markovian dynamics and rewards \cite{Anantharam_Walrand_Markov_reward,Jaksch:2010:NRB:1756006.1859902,Bartlett2009REGALAR,Fruit:2018,yemini2021restless,gafni2022learning}, and  combinatorial bandits with monotone reward functions \cite{Combinatorial_general_nips2016}.
In these models, the prevalent measure for the performance of the player's arm choices is the total expected regret, which measures the cumulative difference between the expected reward of the optimal arm and that of the arms that were sequentially chosen by the player.

 In this paper, we consider a different reward structure, which is based on the  \textit{informativness} of the arm. If we consider each of the $K$ arms as an information source emitting independent and identically  distributed (IID) symbols, a player may have the goal of sampling from the source which is most informative. Clearly, there could be different ways to measure this quantity, and in this work we focus on the natural choice of \textit{Shannon's entropy}. This is motivated both by the standard interpretation of entropy as a measure of uncertainty, the convenient analytic properties of the entropy functional, and its practical applicability in applications such as anomaly detection \cite{entropy_max_HAIDARSHARIF20122543,entropy_max_ICNC_2017,entropy_max_DBLP:conf/mm/HuPBG17,DBLP:journals/entropy/HowediLP20}. We thus henceforth refer to a MAB problem with entropy rewards as \textit{informational} MAB (IMAB). At each round the player observes a random symbol generated from the PMF of its chosen arm. Letting  the true probability of the generated symbol $x$ playing arm $i$ be $p_i(x)$, the instantaneous reward associated with this symbol is its self-information $-\log p_{i}(x)$. Using the symbol observations from the previously played arms, the player aims to choose the arm with maximal entropy. 
 Evidently, this model is different from standard MAB since the expected reward of each arm, to wit, its entropy, is a \textit{non-linear functional} of the PMF, rather than its mean (which is a linear functional). Moreover, at each round $t$, the instantaneous reward function depends on the \textit{probability} of a symbol and not its value $x(t)$. Therefore, the player does not directly observe the instantaneous rewards $-\log p_{i}(x(t))$, but can only estimate it based on its previous observations. As a result, and as we next discuss, the IMAB problem is intimately related to the problem of confidence bounds in entropy estimation.

A highly successful algorithmic approach to MAB problems is \textit{optimism in the face of uncertainty}, where the uncertainty in the reward estimation of each arm is replaced by an optimistic estimate that is based on upper confidence bounds (UCB)
\cite{Auer:2003:UCB:944919.944941,bubeck2012regret}. The performance guarantees of the expected regret of UCB algorithms are achieved by utilizing concentration inequalities (see  primers in \cite{Dubhashi2009,CIT-064_igal}), which  bound the probability that the unbiased sampled mean of the reward is outside a chosen distance, known as the confidence interval, of the expected reward of an arm. For the entropy functional, the plug-in estimator of the entropy is well-known to be \textit{biased}, and in fact, there are no finite variance unbiased estimators of the entropy in general alphabet discrete settings \cite{Paninski2003}. However, it is also known that the bias of the plug-in estimator is upper bounded by $\log(1+\frac{|\cal{X}|}{n})$ where $|\cal X|$ is the alphabet size and $n$ is the number of samples used for estimation \cite{Paninski2003}. Moreover, the entropy functional satisfies a bounded-difference inequality with respect to (w.r.t.) to the samples, and so an application of McDiarmid's inequality  \cite{Antos2001} resulted in  a concentration inequality bound for the plug-in entropy estimator w.r.t. its (biased) mean. Therefore, our first approach for confidence intervals in entropy estimation is to use a bias-corrected plug-in entropy estimator. Nonetheless, the drawback of this approach is that the additional bias term in the confidence interval leads to a large interval in case the alphabet of the arm is large, even if the entropy of this arm is very low. Therefore, we develop a second type of confidence interval 
bounds that is based on a total variation bound on the entropy difference of a pair of PMFs \cite{Ho_Yeung2010, sason2013entropy}. This bound further hinges on a concentration inequality for the total variation, which depends on a  functional of the arm's PMF (denoted $\zeta(p)$ in what follows), which essentially quantifies the effective alphabet size of the arm 
(given by $\zeta(p)|\cal X|$). We then show that $\zeta(p)$ itself can be estimated from the samples at the previous rounds, and this estimate can be used in a UCB algorithm in lieu of the true value. 
Our upper confidence intervals directly depend on the alphabet size. In practice, the alphabet size may not be known to the player in advance and only a loose upper bound on the support size can be available. Therefore, we additionally propose a UCB algorithm that incorporates support size estimation. For the sake of simplicity of presentation, we derive these bounds for the bias-corrected estimator approach, but it can nonetheless be similarly applied to  confidence interval bounds that are based on a total variation bound. 

\paragraph*{Main Contributions and paper outline}
In Sec. \ref{sec:Problem-Formulation} we formulate the IMAB problem, and  in  Sec. \ref{sec:The-Upper-Confidence} we state a generic UCB algorithm for the IMAB problem, which takes
a choice of an entropy estimator and a choice of an upper confidence
bound (on that estimator) as inputs. In later sections we specify this algorithm for particular choices, and derive regret bounds. First, in Sec. \ref{sec: Bias corrected UCB} we bound the regret of a UCB algorithm that is based on 
a bias-corrected plug-in estimator, and obtain a regret upper bound, which roughly scales as the standard UCB bound (for mean-based rewards), yet only after a large number of rounds $O(\exp{\sqrt{|\cal X|}})$ where here $\cal X$ is the maximal alphabet size of the arms. Then, in Sec. \ref{sec: TV UCB} we present a UCB algorithm that is based on concentration of total variation distance, with the goal of ameliorating this dependence on the alphabet size in case the PMFs of the arms are close to the vertices of the simplex (to wit, there exist a letter whose probability is close to $1$). This regime is where elaborated UCB algorithms may lead to improved regret bounds. From a practical point of view, this fits anomaly detection scenarios, in which the arms are mostly "idle", and thus most of the time emit the high probability symbol, and only occasionally a different symbol ("anomaly"). The player then needs to find the arm which is the "least idle". The UCB algorithms in this section are based on  data-dependent confidence intervals, similarly to variance-UCB, which has the merit of adapting the bound to cases in which the entropy of the source is much smaller compared to its maximal value. While our motivation is the large alphabet case, in order to facilitate ideas in a clean way, we first consider Bernoulli arms, for which the probability of the symbol $'1'$ being close to zero indicates that the source is mostly idle. In addition, this setting can also
be compared with the Lai-Robbins lower bound \cite[Thm. 1]{lai1985asymptotically}, which reveals the asymptotic optimality of the proposed UCB algorithm. We then extend the analysis to alphabets of arbitrary size. 
Afterward, in Sec. \ref{sec:unknown_support} we relax the assumption that the support sizes $|\mathcal{X}_{i}|$ are known to the player, and  assume a bound of $\kappa^{-1}$ on the  minimal probability in the support. This leads to an upper bound of $\kappa$ on the support size, which may significantly overestimate the true support size. In turn, this leads to an overestimate of the bias of the entropy estimator, then to an overestimate of the confidence interval used by the player, and consequently, to an excessively large regret. To ameliorate this phenomenon, we propose a confidence interval bound that is based on online support size estimation, and analyze the corresponding regret. The proposed algorithm leads to improved regret bound that can be smaller by a factor of $\sqrt{\kappa}$ compared to using the naive bias-corrected plug-in estimator (with $\kappa$ used for the alphabet size). 
Finally, in Sec. \ref{sec:Simulations} we provide a few  numerical examples that support the theoretical findings, and in 
Sec. \ref{sec:Summary-and-Future} we summarize the paper. 

\paragraph*{Related work} We conclude the introduction by mentioning related work in the entropy estimation and bandit problem literature. 
The general problem of entropy estimation is well studied \cite{Antos2001,Paninski2003,Grassberger2008,Ho_Yeung2010,Nemenman2011,Chao2013,Archer2013,Valiant2013,Schurmann2015,Jiao2015,Jiao2017}. These papers lead to tight (and even optimal) entropy estimators, and here we build upon their ideas to obtain a confidence interval bound, which is both tailored to the IMAB problem and can also be efficiently estimated from data. 
In the multi-armed bandit literature, information-theoretic functionals  have been used in recent years  
to decrease the expected regret of several MAB models 
\cite{NIPS2014_IDS_Van_Roy,pmlr-v75-kirschner18a,pmlr-v134-kirschner21a}. Using the mutual information of the probabilities for arm sampling at two consecutive rounds, information-directed sampling (IDS) outperforms both UCB-based algorithms and Thompson sampling \cite{THOMPSON_1933,Thompson_JMLR:v17:14-087}  in problems with special structures, such as dependent on prior models \cite{NIPS2014_IDS_Van_Roy}, and bandits with  arm-dependent heteroscedastic noise \cite{pmlr-v75-kirschner18a}. Nonetheless, these works utilize informational measures to create exploration-exploitation trade-offs, and the reward structure is standard. Extending reward estimation for reward functions whose mean depends on the higher moments, or even the complete knowledge of distribution function is the focus on the works  \cite{Mannor_2014_general_reward,Combinatorial_general_nips2016}. However, the work \cite{Mannor_2014_general_reward} is limited to a known parametric family of distributions with unknown parameters. Moreover, \cite{Combinatorial_general_nips2016} relies on  stochastically dominant confidence bounds that requires monotonically increasing instantaneous reward function, however the self-information $-\log(p_i(x))$ is monotonically \textit{decreasing} in $p_i(x)$. 
Furthermore, it is assumed in  \cite{Mannor_2014_general_reward,Combinatorial_general_nips2016} that the instantaneous reward is directly known, however, in our case the instantaneous reward is unknown to the player and is not a function of the symbol outcome but rather of its probability. 

\section{Problem Formulation \label{sec:Problem-Formulation}}
We first define a few notation conventions that will be used in the rest of the paper. For $a,b\in\mathbb{R}$, we denote $\max\{a,b\}\dfn a\vee b$ and $\min\{a,b\}\dfn a\wedge b$, as well as $(t)_{+}=t \vee0$. Furthermore, we denote by $\boldsymbol{X}\cup\boldsymbol{Y}$ the concatenation of the vector $\boldsymbol{X}$ and the vector $\boldsymbol{Y}$.
To focus the reader on the first-order terms  
we denote the linear-times-polylogarithmic function by
\begin{equation}
\Lambda_{k}(s)\dfn s\log^{k}s,\label{eq: linear-times-polylog}
\end{equation}
(which can be thought of as an almost linear function). 
For a discrete alphabet $\cal Y$, we denote the entropy of a PMF $p$ over an alphabet $\cal Y$  by $H(p)\dfn-\sum_{y\in{\cal Y}}p(y)\log p(y)$,
where logarithms are arbitrarily taken to the natural base. We denote the \textit{total variation distance} between two PMFs $p$ and $q$ on a finite alphabet
${\cal Y}$ by $\dtv(p,q)\dfn\sum_{y\in{\cal Y}}\left|p(y)-q(y)\right|$.
We remark at this point that in what follows, as customary, we have opted for  the simplicity of our bounds over obtaining the tightest constants possible.

Consider the following IMAB problem. Let
$\{X_{i}\}_{i=1}^{K}$ be a set of $K\geq2$ memoryless sources, each
defined on a possibly different alphabet ${\cal X}_{i}$, such that
$p_{i}(x)\dfn\P[X_{i}=x]$. We further denote by $p_{i}=\{p_{i}(x)\}_{x\in{\cal X}_{i}}$
the full PMF of the $i$th source. The IMAB problem is a game in which at each
round $t$, the player chooses one of the sources $i\in[K]\dfn\{1,2\ldots,K\}$
and observes the $t$-th symbol $X_{i}(t)$
from that source. In the context of MAB, each of the sources is referred
to as an \emph{arm}. In this paper, we assume that the random reward associated
with this arm choice and this observation is the \emph{self-information}
$-\log p_{i}(X_{i}(t))$, and so the expected reward of sampling only
arm $i$ is $-\E[\log p_{i}(X)]=H(p_{i})\dfn H_{i}$,
which is the entropy of the $i$th source. The goal of the player is to choose the arm with the maximal expected
reward, that is, the maximal entropy, $i^{*}\in\argmax_{i\in[K]}H_{i}$.
The player, which does not know in advance the PMFs $p_{i}$
(and so also not the entropy values $H_{i}$) estimates the expected
reward $H_{i}$ of each arm from its previous actions and observations. The first part of the paper, i.e., Sections \ref{sec: Bias corrected UCB}-\ref{sec: TV UCB} assumes that the player knows in advance the alphabet size of the sampled random variables. Later on, in Section \ref{sec:unknown_support}, the scenario in which the player does not know in advance the exact alphabet size is considered, and instead it only has access to a loose upper bound on it. This models the scenario in which the alphabet size is considerably larger than the actual support size.

We denote the arm choice of the player at round $t$ by $I(t)$, and
we let $N_{i}(t)=\sum_{\tau=1}^{t}\I[I(\tau)=i]$ be the number of
times in which arm $i$ was sampled up to round $t$. To measure the performance of the policies used by the player, we will adopt the standard expected \textit{pseudo-regret} \cite[Ch. 1]{bubeck2012regret}
\begin{equation}
R(t)\dfn t\cdot H_{i^{*}}-\sum_{i\in[K]}\E(N_{i}(t))\cdot H_{i}.\label{eq:exp_regert_def}
\end{equation}
Letting $\Delta_{i}\dfn H_{i^{*}}-H_{i}$ denote
the \textit{gap} of the $i$th arm, we may equivalently represent the expected
pseudo-regret as 
\begin{equation}
R(t)=\sum_{i\in[K]:\Delta_{i}>0}\E(N_{i}(t))\cdot \Delta_{i}.\label{eq: expected pseudo-regret}
\end{equation}

\section{The Upper Confidence Bound Algorithm for Entropy Rewards \label{sec:The-Upper-Confidence}}

In this section, we present a generic UCB algorithm for the IMAB problem.
Similarly to the UCB algorithm with standard rewards \cite[Sec. 2.2]{bubeck2012regret}
\cite[Ch. 1]{slivkins2019introduction}, the algorithm is based on
an entropy estimator for which an upper confidence bound is known
to hold with high probability. In general, let $\boldsymbol{Y}\dfn\{Y_{\ell}\}_{\ell\in[n]}$
be $n$ IID samples from a PMF $p$ over a finite alphabet ${\cal Y}$.
Suppose that there exists an entropy estimator $\hat{H}(\boldsymbol{Y},n)$
and an upper confidence deviation (UCD) function $\ucb(\boldsymbol{Y},n,\delta)$
(for $\delta\in(0,1)$) for which the upper confidence bound 
\begin{equation}
H(p)\leq\hat{H}(\boldsymbol{Y},n)+\ucb(\boldsymbol{Y},\delta,n)\label{eq: UCB for entropy general}
\end{equation}
holds with probability larger than $1-\delta$. Note that both the
estimator $\hat{H}(\boldsymbol{Y},n)$ and the confidence deviation
$\ucb(\boldsymbol{Y},\delta,n)$ may depend on the observed source
symbols $\boldsymbol{Y}$. The algorithm keeps a set of observed samples
from each of the arms up to any round $t$. Based on the samples of
each arm, the algorithm computes the value of the estimator and the
upper confidence deviation for each of the arms. The played arm is
then the one maximizing the estimated entropy plus the confidence
deviation, that is, the right-hand side (RHS) of (\ref{eq: UCB for entropy general})
for the set of observations of each of the arms. The new observed
sample is then added to the set of observations of that played arm.

The algorithm takes as input the following:
\begin{itemize}
\item The parameters of the information sources, namely, the number of arms
$K$ and the alphabet sizes $\{{\cal X}_{i}\}_{i\in[K]}$. When these exact values are not known in advance, with a slight abuse of notations, we use the input parameters  $\{\kappa_{i}\}_{i\in[K]}$ such that $|{\cal X}_{i}|\leq\kappa_{i}$ for all $i\in[K]$. These inputs are used in Section \ref{sec:unknown_support} which considers the unknown alphabet case;
\item A sequence of entropy estimators $\{\hat{H}(\cdot,n)\}_{n\in\mathbb{N}_{+}}$
and a sequence of upper confidence deviation functions $\{\ucb(\cdot,\cdot,n)\}_{n\in\mathbb{N}_{+}}$;
\item A real confidence parameter $\alpha>2$ and a confidence function
$\delta(t)\equiv\delta_{\alpha}(t)$ which determines the required
reliability of the confidence interval at any round $t$. 
\end{itemize}
At each round $t\in\mathbb{N}_{+}$, the algorithm plays a chosen
arm $i\in[K]$ and observes the sample $X_{i}(t)$. The output of
the algorithm is $\{N_{i}(t)\}_{i\in[K],\;t\in\mathbb{N}_{+}}$ the
number of times each of the arm have been played up to each of the
rounds $t$ (or, equivalently, the played arm at each round $\{I(t)\}_{t\in\mathbb{N}_{+}}$).
Given this output, the pseudo-regret at round $t$ is given by $\sum_{i\in[K]:\Delta_{i}>0}N_{i}(t)\Delta_{i}$
whose expected value is $R(t)$, as in (\ref{eq: expected pseudo-regret}).
The input, the actions and the policy of the player are summarized in Algorithm
\ref{alg:A-UCB-general}. Therein, $\boldsymbol{X}_{i}(t)$ is the
set of samples available to the player at round $t$ from the $i$th
arm. 

\begin{algorithm}
\begin{algorithmic}[1]

\Procedure{An Upper Confidence Bound Algorithm}{$K$, $\{{\cal X}_{i}\}_{i\in[K]}$,
$\hat{H}(\cdot,n)$ ,$\ucb(\cdot,\cdot,n)$, $\alpha$, $\delta_{\alpha}(t)$}

\State \textbf{set} $\boldsymbol{X}_{i}(0)=\phi$ and $N_{i}(0)=0$
for all $i\in[K]$ 

\Comment{ The observation set of each arms is empty at round $t=0$}

\For{ $t=1,2,\ldots$}

\State  \textbf{play} $I(t)\in\arg\max_{i\in[K]}\{\hat{H}(\boldsymbol{X}_{i}(t-1),N_{i}(t-1)))+\ucb(\boldsymbol{X}_{i}(t-1),\delta_{\alpha}(t),N_{i}(t-1))\}$

\State  \textbf{set} $\boldsymbol{X}_{I(t)}(t)=\boldsymbol{X}_{I(t)}(t-1)\cup X_{I(t)}(t)$
and $N_{I(t)}(t)=N_{I(t)}(t-1)+1$

\Comment{ The observation of the chosen arm is concatenated to the sequence of
observations}

\State  \textbf{set} $\boldsymbol{X}_{i}(t)=\boldsymbol{X}_{i}(t-1)$
and $N_{i}(t)=N_{i}(t-1)$ for all $i\in[K]\backslash I(t)$

\Comment{The observation set of others arms is unchanged}

\EndFor

\State \textbf{return} $\{N_{i}(t)\}_{i\in[K],\;t\in\mathbb{N}_{+}}$ 

\Comment{The number of times each arm $i\in[K]$ have been played
up to each round $t\in\mathbb{N}_{+}$}

\EndProcedure

\end{algorithmic}

\caption{A general UCB-entropy algorithm \label{alg:A-UCB-general}}
\end{algorithm}

Table \ref{table:notations_def_used_thm} summarizes the entropy estimators $\hat{H}(\cdot,n)$, upper confidence bounds $\text{UCB}(\cdot,\cdot,n)$, and functions $\delta_{\alpha}(t)$ that we develop for Algorithm \ref{alg:A-UCB-general}. Additionally,  Table \ref{table:notations_def_used_thm} refers to the relevant theorems that provide performance guarantees for the chosen inputs. 

\begin{table}[h!]
\centering
\begin{tabular}{||c |c| c| c|c|c|c|c||} 
 \hline
  Thm. &Alphabet &  $\hat{H}(\boldsymbol{Y},n)$ & $\delta_{\alpha}(t)$ & $\text{UCB}(\boldsymbol{Y},\delta,n)$ & \shortstack{PMF based\\ UCB} & \shortstack{Pseudo\\regret}   \\ [0.5ex]
 \hline\hline
Thm. \ref{thm: UCB-Bias regret} & \shortstack{discrete, \\ finite,\\known} & $H(\hat{p}(n))$ & $t^{-\alpha}$&  \eqref{eq: confidence bound bias} & no &  \eqref{eq: regret of bias}  \\ \hline
\rule{0pt}{15pt}Thm. \ref{thm: UCB-Bernoulli regret}& \shortstack{binary,\\known}& $H(\hat{p}(n))$ & $6t^{-\alpha}$ & \eqref{eq: confidence bound Ber} & yes & \eqref{eq: regret bound Ber}\\ \hline
\rule{0pt}{17pt}Thm. \ref{thm: UCB-Bernoulli regret close to half}& \shortstack{binary,\\known} & $H(\hat{p}(n))$ &$4t^{-\alpha}$&  \eqref{eq: confidence bound Ber half}& yes & \eqref{eq: regret bound Ber close to half}\\ \hline
Thm. \ref{thm: UCB-TV regret} & \shortstack{discrete, \\finite,\\known} & $H(\hat{p}(n))$ &$t^{-\alpha}$&  \eqref{eq: confidence bound TV}& yes & \eqref{eq: u tv}\\ \hline
Thm. \ref{thm:upper_regret_uknown_support}& \shortstack{finite but \\ unknown support} & $H(\hat{p}(n))$& $t^{-\alpha}$ & \eqref{eq:ucb_uknown_support}& no & \eqref{eq:upper_pregret_uknown_support}\\
[1ex]
 \hline
\end{tabular}
\vspace{0.1cm}
\caption{Summary of inputs for Algorithm \ref{alg:A-UCB-general} and results. }
\label{table:notations_def_used_thm}
\end{table}

\section{Upper Confidence Bounds with Bias-Corrected Entropy Estimation\label{sec: Bias corrected UCB}}

A straightforward idea for estimating entropy is the plug-in estimator,
in which the PMF of the source is estimated via the empirical PMF
of the samples, and then the entropy of the empirical PMF is used
to estimate the entropy of the source. As discussed in the introduction, the plug-in
estimator for the entropy concentrates around its expected value \cite{Antos2001},
yet suffers from a negative bias \cite{Paninski2003}. Thus, a natural
method of obtaining an upper confidence bound is by correcting this bias. Specifically, let $\boldsymbol{Y}=\{Y_{\ell}\}_{\ell\in[n]}$
be a sequence of IID samples from some distribution $p$ over the
alphabet ${\cal Y}$, and let (with a slight abuse of notation) $\hat{p}(n)=\{\hat{p}(y,n)\}_{y\in{\cal Y}}$
be the empirical mean of the $n$ samples, where $\hat{p}(y,n)\dfn\frac{1}{n}\sum_{\ell=1}^{n}\I\{Y_{\ell}=y\}$
for all $y\in{\cal Y}$.  
Then, the plug-in estimator $H(\hat{p}(n))$
is biased, and, as was proved in \cite{Paninski2003}, 
\begin{equation}
H(p)-B(n)\leq\E[H(\hat{p}(n))]\leq H(p),\label{eq:bias_negative_lower_bound}
\end{equation}
where 
\begin{equation}
B(n)\dfn\log\left(1+\frac{|{\cal Y}|-1}{n}\right)
\end{equation}
for $n\geq1$.
Therefore, the bias-corrected estimator $H(\hat{p}(n))+B(n)$ has
a nonnegative bias. Let 
\begin{flalign}\label{eq: confidence bound bias}
\ucb_{\text{bias}}(\delta,n)\dfn B(n)+\sqrt{\frac{2\log^{2}(n)}{n}\log\left(\frac{2}{\delta}\right)}.
\end{flalign}
The concentration result of the plug-in estimator from \cite[p. 168]{Antos2001}
implies the following confidence interval bound. For the sake of clarity of exposition, we present all the relevant proofs for this section in Appendix \ref{append:proofs sec Bias corrected UCB}.
\begin{prop}
\label{prop: bias corrected plug-in estimator}Let $\boldsymbol{Y}=\{Y_{\ell}\}_{\ell\in[n]}$
be IID from a discrete distribution $p$ over a \emph{finite} alphabet
${\cal Y}$ such that $p(y)\dfn\P[Y=y]$. Then, assuming $n\geq2$,
it holds for any $\delta\in(0,1)$ that
\[
\left|H(\hat{p}(n))-H(p)\right|\leq\ucb_{{\text{\emph{bias}}}}(\delta,n),
\]
with probability larger than $1-\delta$. 
\end{prop}

We may now specify the general Algorithm \ref{alg:A-UCB-general}
to the upper confidence bound of Proposition \ref{prop: bias corrected plug-in estimator},
and obtain the following guarantee on the expected regret. To this
end, let us denote
\begin{equation}
\Gamma_{\text{bias}}(\alpha,\beta,{\cal Y},\Delta,t)\dfn\max\left\{ \frac{|{\cal Y}|-1}{e^{\beta\cdot\Delta/2}-1},\;15\cdot\Lambda_{2}\left(\frac{8\cdot\log(2t^{\alpha})}{(1-\beta)^{2}\Delta^{2}}\right)\right\} .\label{eq: u bias}
\end{equation}

\begin{thm}
\label{thm: UCB-Bias regret}Assume that Algorithm \ref{alg:A-UCB-general}
is run with a plug-in entropy estimator $\hat{H}(\boldsymbol{Y},n)\equiv H(\hat{p}(n))$,
and upper confidence deviation $\ucb(\boldsymbol{Y},\delta,n)\equiv\ucb_{\text{\emph{bias}}}(\delta,n)$
with $\delta\equiv\delta_{\alpha}(t)=t^{-\alpha}$ and $\alpha>2$.
Let $\beta\in(0,1)$ be given. Then, the pseudo-regret is bounded
as 
\begin{equation}
R(t)\leq\sum_{i\in[K]:\Delta_{i}>0}\left[\Gamma_{\text{\emph{bias}}}(\alpha,\beta,{\cal X}_{i},\Delta_{i},t)\cdot\Delta_{i}+\frac{2(\alpha-1)}{\alpha-2}\cdot\Delta_{i}\right].\label{eq: regret of bias}
\end{equation}
\end{thm}

We remark that $\beta$ is a parameter that is used in our analysis of the pseudo regret for the sake of serving tighter upper bounds and is not part of Algorithm \ref{alg:A-UCB-general}. The bound on the regret of Algorithm \ref{alg:A-UCB-general}
with a bias corrected entropy estimator in (\ref{eq: regret of bias}) is comprised
of a few terms. However, for any $i\in[K]$, there is only a single term that blows-up
as $\Delta_{i}\downarrow0$, given by
\[
\frac{c_{1}\cdot\log(t)}{\Delta_{i}}\cdot\log^{2}\left(\frac{c_{2}\cdot\log(t)}{\Delta_{i}^{2}}\right),
\]
for some constants $c_{1},c_{2}$. Thus, the regret scales as $\tilde{O}(\frac{\log(t)}{\Delta_{i}})$,
where the only difference from the standard UCB \cite[Thm. 2.1]{bubeck2012regret}
is the additional poly-logarithmic term. Then, if we consider for
simplicity the two-arm case ($K=2$) with $\Delta_{1}=0$ and $\Delta_{2}\equiv\Delta$,
since $\Delta t$ is always an upper bound on the pseudo-regret, we
may obtain the $\tilde{O}(\frac{1}{\Delta}\wedge \Delta t)=\tilde{O}(\sqrt{t})$,
which roughly matches the gap-independent bound in the standard MAB problem 
(e.g., \cite[Thm. 2.10]{slivkins2019introduction}).

Nonetheless, from a different perspective, assuming that the gaps are
all constants, then if $\log^{2}(t)=\tilde{O}(|{\cal X}_{i}|)$ the
regret will be determined by the first term in (\ref{eq: u bias}),
and so the regret bound is large as long as $t=O(\max_{i\in[K]}\exp(\sqrt{|{\cal X}_{i}|}))$. In the next section we develop a UCB algorithm which ameliorates this unfavorable behavior. 

\section{Upper Confidence Bounds with a Total Variation Bound \label{sec: TV UCB}}

As we have seen in Theorem \ref{thm: UCB-Bias regret}, the upper
bound on the regret of the UCB algorithm with a bias corrected entropy estimator is severely affected by the size of the alphabets ${\cal X}_{i}$. 
Therefore, a natural question is whether improved bounds can be obtained whenever the entropy of sources is much less than the alphabet size.
In this section, we propose algorithms that adapt to arms with very low entropy. The idea is similar to
variance-UCB \cite{auer2002finite,audibert2009exploration} that replaces the distribution-independent confidence interval of the standard UCB algorithm (which hinges, e.g., on  Hoeffding's inequality, assuming bounded rewards), with a  distribution-dependent confidence interval (which hinges, e.g., on Bernstein's inequality).
Our next proposed algorithms will similarly use a data-dependent UCD. For such algorithms, the confidence interval, which in principle depends on the unknown distribution, is also required to be estimated from the given observations. For the sake
of illustration, let us first consider the simpler case of Bernoulli arms, for which $p_{i}(1)=\P[X_{i}=1]$ is close to $0$ for some arm $i\in[K]$. The entropy of this arm is much smaller than the maximal possible
value of $\log|{\cal X}_{i}|=\log2$. A multiplicative
Chernoff's inequality (or Bernstein's inequality, see Lemma \ref{lem: empirical mean concentration Ber})
results a confidence interval of $O(\sqrt{\frac{p_{i}(1)\log(1/\delta)}{n}})$
in the estimation of $p_{i}(1)$ using $n$ samples from the source.
Since $p_{i}(1)\ll1$, this is much smaller than the $O(\sqrt{\frac{\log(1/\delta)}{n}})$
which stems from standard Chernoff's bound (or Hoeffding's inequality).
This confidence interval on $p_{i}(1)$ then leads to an improved
confidence interval bound on the error of the plug-in estimator of
the entropy. Since this confidence interval bound depends
on the unknown $p_{i}(1)$, it should also be estimated by the player, using its estimation of $p_{i}(1)$. The estimation error of the confidence interval is then another source of error that is addressed by our analysis. 

Thus, in what follows, we begin with the Bernoulli
case in Sec. \ref{subsec:The-Bernoulli-Case}, which leads to a more
transparent bound than the general case, and can also be compared
to the Lai-Robbins impossibility result \cite[Thm. 1]{lai1985asymptotically}.
We later on generalize this type of algorithm to arbitrary arm alphabets
in Sec. \ref{subsec:The-General-Alphabet}. As in the Bernoulli case, the confidence interval of 
arms with low entropy is smaller than ones with large entropy. The PMF of these arms is close to the vertices of the probability simplex, which in the Bernoulli case implies a low value of $p_{i}(1)$. In the general alphabet case, the value of $p_i(1)$ is replaced by a functional\footnote{We define this functional explicitly in \eqref{eq:zeta_def}.} $\zeta(p)\in[0,1]$, 
which satisfies that low values of $\zeta(p)$ are indicator of being close to the vertices of the simplex, 
and that can also be efficiently be estimated from the data. As a result, we additionally show that the effective alphabet size for the IMAB problem is $\zeta_{i}|{\cal X}_{i}|$, which demonstrates the utility of this functional.

\subsection{The Bernoulli Case \label{subsec:The-Bernoulli-Case}}

In this section, we consider the Bernoulli case, in which ${\cal X}_{i}=\{0,1\}$
for all the $K$ arms, $i\in[K]$. 
For brevity we use $h_{b}(p)\dfn-p\log p-(1-p)\log(1-p)$
to denote the binary entropy function. Furthermore, we assume for
simplicity of exposition that $p_{i}(1)=\P[X_{i}=1]\leq1/2$ for all
$i\in[K]$. The results can be extended in a straightforward manner
to remove this assumption. 
The proofs of the theoretical results presented in Section \ref{subsec:The-Bernoulli-Case} are included in  Appendix \ref{append:proofs subsec The-Bernoulli-Case}.
The proposed UCB algorithm and its regret
analysis are based on the following confidence deviation function
\begin{equation}
\ucb_{\text{ber}}(q,\delta,n)\dfn\sqrt{\frac{12q\log(\frac{6}{\delta})}{n}}\log\left(\frac{n}{q\log(\frac{6}{\delta})}\right)+\frac{18\log(\frac{6}{\delta})\log(n)}{n},\label{eq: confidence bound Ber}
\end{equation}
and the corresponding confidence interval bound for the plug-in entropy
estimator:
\begin{prop}
\label{prop: Bernoulii plug-in estimator}Let $\boldsymbol{Y}=\{Y_{\ell}\}_{\ell\in[n]}$
be IID from a Bernoulli with parameter $p=\P[Y_{i}=1]$, and let $\hat{p}(n)=\frac{1}{n}\sum_{\ell=1}^{n}\I\{Y_{\ell}=1\}$
be the empirical probability of $'1'$. Let $\delta\in[0,\frac{1}{2}]$
be given. If $n\geq200\cdot\log(\frac{4}{\delta})$ then 
\[
\left|h_{b}(\hat{p}(n))-h_{b}(p)\right|\leq\ucb_{\text{\emph{ber}}}(\hat{p}(n),\delta,n),
\]
with probability larger than $1-\delta$. 
\end{prop}

\begin{rem}
The confidence interval of Proposition \ref{prop: Bernoulii plug-in estimator}
follows from the relation 
\begin{equation}
\left|H(p)-H(q)\right|\leq\dtv(p,q)\log\left(\frac{|{\cal Y}|}{\dtv(p,q)}\right),\label{eq: entropy difference bound via total variation simple}
\end{equation}
where $\dtv(p,q)$ is the total variation distance between PMFs $p$
and $q$ defined on a common alphabet ${\cal Y}$, and that holds
as long as $\dtv(p,q)\leq\frac{1}{2}$ \cite[Lemma 2.7]{csiszar2011information}.
The bound (\ref{eq: entropy difference bound via total variation simple})
is not the sharpest known bound, and, e.g., it also holds that \cite[Thm. 6]{Ho_Yeung2010}
\begin{equation}
\left|H(p)-H(q)\right|\leq\dtv(p,q)\log\left(|{\cal Y}|-1\right)+h_{b}\left(\dtv(p,q)\right),\label{eq: entropy difference bound via local and total variation}
\end{equation}
(see also an additional refinement in \cite{sason2013entropy}). In
our proofs this type of bounds is utilized in the regime $\dtv(p,q)=o(1)$,
for which both (\ref{eq: entropy difference bound via total variation simple})
and (\ref{eq: entropy difference bound via local and total variation})
are of the same order of $\Theta\left(\dtv(p,q)\log\frac{|{\cal Y}|}{\dtv(p,q)}\right)$.
Thus, we exclusively use the simpler bound (\ref{eq: entropy difference bound via total variation simple}). 
\end{rem}

\begin{rem}
In the special case of a binary alphabet, the binary entropy $h_b(p)$ is maximized when $p=1/2$. Additionally, it is symmetric around $p=1/2$ and strictly increases in the interval $[0,1/2]$.
Therefore, instead of utilizing UBC for entropy estimation we can consider the following strategy where we optimistically look for the arm with the minimal $|p-1/2|$ in the set $p\in[\hat{p}_i(1)- \ucb,\hat{p}_i(1)+\ucb]$.  
However, since we examine the binary alphabet to gain insights regarding the general alphabet case, we focus below on regret bounds which we can extend to the general alphabet case. These regret bounds are achieved by UCB strategies for entropy estimation.
\end{rem}

Next, we state the regret bound on Algorithm \ref{alg:A-UCB-general},
based on the confidence interval of Proposition \ref{prop: Bernoulii plug-in estimator}.
To this end, let us denote
\begin{multline}
\Gamma_{\text{ber}}(\alpha,\beta,q,\Delta,t)\dfn\\
\max\left\{ 6\cdot\Lambda_{1}\left(\frac{36\alpha\log(t)}{(1-\beta)\Delta}\right),\frac{5120q\alpha\log(t)}{\beta^{2}\Delta^{2}}\cdot\log^{2}\left(\frac{48}{\beta^{2}\Delta^{2}}\right),\frac{88\sqrt{\alpha\log(t)}}{\beta\Delta}\cdot\log\left(\frac{48}{\beta^{2}\Delta^{2}}\right)\right\} \label{eq: u ber},
\end{multline}
and use the notation $\hat{p}(\boldsymbol{Y},n)$ for the empirical
probability of $'1'$ in $\boldsymbol{Y}=\{Y_{\ell}\}_{\ell\in[n]}$. 
\begin{thm}
\label{thm: UCB-Bernoulli regret}Assume that ${\cal X}_{i}=\{0,1\}$
for all $i\in[K]$. Further assume that Algorithm \ref{alg:A-UCB-general}
is run with the plug-in entropy estimator $\hat{H}(\boldsymbol{Y},n)\equiv H(\hat{p}(\boldsymbol{Y},n))$
and upper confidence deviation 
\[
\ucb(\boldsymbol{Y},\delta,n)\equiv\ucb_{\text{\emph{ber}}}(\hat{p}(\boldsymbol{Y},n),\delta,n),
\]
(as defined in (\ref{eq: confidence bound Ber})) with $\delta\equiv\delta_{\alpha}(t)=6t^{-\alpha}$
with $\alpha>2$. Then, 
\begin{equation}
R(t)\leq\sum_{i\in[K]:\Delta_{i}>0}\inf_{\beta\in(0,1)}\left[\Gamma_{\text{\emph{ber}}}(\alpha,\beta,p_{i}(1),\Delta_{i},t)\cdot\Delta_{i}+\frac{16(\alpha-1)}{\alpha-2}\cdot\Delta_{i}\right].\label{eq: regret bound Ber}
\end{equation}
\end{thm}

Let us inspect the regret bound of (\ref{eq: regret bound Ber}) of
Theorem \ref{thm: UCB-Bernoulli regret} in the regime of small gaps.
By inserting the definition of the $\Gamma_{\text{ber}}(\cdot)$,
the dominating term as $\Delta_{i}\downarrow0$ is on the order of
$\tilde{O}(\frac{p_{i}\log(t)}{\Delta_{i}})$ and all other terms
are $O(\log^{c}(\frac{1}{\Delta_{i}}))$. Thus, e.g., in case $\Delta_{i}=\Theta(p_{i})$
(as $t\to\infty$), then the regret is only $O(\log(t)\cdot\log^{c}(\frac{1}{\Delta_{i}}))$.
This is a similar behavior to the variance-UCB algorithm \cite{auer2002finite,audibert2009exploration}
with standard bounded rewards. Nonetheless, in general, the bound of Theorem \ref{thm: UCB-Bernoulli regret} is not optimal. Indeed, recall that in the standard Bernoulli MAB problem,  say with two arms ($K=2$), the regret bound depends on the difference $p_1(1)-p_2(1)$ between the $'1'$-probability of each arm, which is exactly the gap between their rewards. However, in the IMAB problem, the gap is $h_b(p_1(1))-h_b(p_2(1))$, and due to the non-linearity of the entropy functional this gap depends on both the difference $p_1(1)-p_2(1)$ as well as the location of $p_1(1)$. 

To further elucidate this phenomenon, next we state
the Lai-Robbins lower bound \cite{lai1985asymptotically} on the pseudo-regret. We follow the clear statement made
in \cite[Thms. 2.14 and 2.16]{slivkins2019introduction}:
\begin{thm}[Lai-Robbins lower bound]
\label{thm:An asymptotic lower bound} Consider the IMAB problem with
$K$ arms. A problem instance ${\cal I}$ is the collection $\{p_{i}\}_{i\in[K]}$
with $p_{i}\equiv p_{i}(1)\in[0,1/2)$. Suppose that a IMAB algorithm
is such that $R(t)=O(C_{{\cal I},a}t^{a})$ for each problem instance
$I$ and $a>0$. Fix an arbitrary problem instance ${\cal I}$. Then,
\[
\liminf_{t\to\infty}\frac{R(t)}{\log(t)}\geq\sum_{i\in[K]\colon\Delta_{i}>0}\frac{\Delta_{i}}{\Dkl(p_{i}||p_{i^{*}})},
\]
where $\Dkl(p||q)\dfn p\log(p/q)+(1-p)\log((1-p)/(1-q))$ is the binary
Kullback-Leibler divergence, and $\Delta_i=\max_{j\in[K]}h_b(p_j)-h_b(p_i)$.\footnote{If $q=0$ or $q=1$ and $p\neq q$ then $\Dkl(p||q)\dfn\infty$. }
\end{thm}
The proof of Theorem \ref{thm:An asymptotic lower bound} for the
IMAB problem is omitted since it is essentially  identical to the proof
of the standard Lai-Robbins lower bound for Bernoulli bandits, which 
can be found in \cite[Thm. 1]{lai1985asymptotically}\cite[Thm. 2.2]{bubeck2012regret}\cite[Ch. 2]{slivkins2019introduction}.

Strictly speaking, Theorem \ref{thm:An asymptotic lower bound} is asymptotic and does not specify the minimal $t$ required for its validity, and is also valid for a \textit{fixed} set of gaps. Nonetheless, we next informally compare the order of convergence it implies with the one attained in Theorem \ref{thm: UCB-Bernoulli regret} while considering the effect of varying the gap. To simplify
the next discussion, we will assume $K=2$ with $p_{2}<p_{1}<\frac{1}{2}$.
Let $\Delta\equiv\Delta_{2}=h_{b}(p_{1})-h_{b}(p_{2})$. In the standard
bandit case, one approximates $\Dkl(p_{2}||p_{1})=\Theta(\Delta^{2})$
(e.g., using Pinsker's inequality and its reverse version), and then
the lower bound is $\Omega(\frac{1}{\Delta}\log(t))$. This lower bound is roughly achieved by the basic UCB algorithm \cite[Thm. 2.1]{bubeck2012regret}
(and the variance-UCB algorithm leads to data-dependent improvement
in the constant). Before comparing this result to the upper bound of Theorem  \ref{thm: UCB-Bernoulli regret}, we note that the latter has an extra multiplicative logarithmic factor, which
can be as large as $\Theta(\log(\frac{\log(t)}{\Delta}))$. To focus on the first-order terms in the regret bound, we next ignore these additional factors in the discussion. We next consider a few different regimes.

To begin, let us assume that $p_{1}=p$ and $p_{2}=p-\Lambda$ with $p$ fixed and
$\Lambda\downarrow0$, then $\Delta=h_{b}(p_{1})-h_{b}(p_{2})=\Theta(\Lambda)$
and $\Dkl(p_{2}||p_{1})=\Theta(\Lambda^{2})$, and the ratio in the
lower bound is $\Theta(\frac{\log t}{\Lambda})$ as in standard bandits.
This roughly matches the upper bound of Theorem \ref{thm: UCB-Bernoulli regret}
on the pseudo-regret achieved by the algorithm, and no significant improvements are anticipated.

Next, we consider the regime in which the probabilities of the arms are close to $1/2$. The binary entropy function "flattens" in this region, and is markedly different from the standard linear reward function. Thus, on an intuitive level, this is not a difficult instance of the problem. More explicitly, let us assume that both  $p_{1}=\frac{1}{2}-\Lambda$ and $p_{2}=\frac{1}{2}-2\Lambda$. Then, both $\Delta=h_{b}(p_{1})-h_{b}(p_{2})=\Theta(\Lambda^{2})$
and $\Dkl(p_{2}||p_{1})=\Theta(\Lambda^{2})$, and the ratio in the
lower bound is asymptotically $\Theta(\log t)$ even if $\Lambda\downarrow0$ and
so also $\Delta\downarrow0$. In this regime, the regret of Algorithm \ref{alg:A-UCB-general}
is upper bounded as $\tilde{O}(\frac{\log t}{\Delta})$ whereas the lower bound is only $\Theta(\log t)$, and so the bounds do not match. 
However, as we next show, the fact that the binary entropy function at $1/2$ is quadratic leads to an ameliorated confidence interval bound. Specifically, consider the following upper confidence deviation 
\begin{equation}
\ucb_{\text{ber}}^{(1/2)}(q,\delta,n)\dfn7\left|\frac{1}{2}-q\right|\cdot\sqrt{\frac{\log(\frac{4}{\delta})}{n}}+\frac{9\log(\frac{4}{\delta})}{n}.\label{eq: confidence bound Ber half}
\end{equation}
We now have the following confidence interval bound.
\begin{prop}
\label{prop: Bernoulii plug-in estimator close to half}Let $\boldsymbol{Y}=\{Y_{\ell}\}_{\ell\in[n]}$
be IID from a Bernoulli with parameter $p=\P[Y_{i}=1]$, let $\hat{p}(n)=\frac{1}{n}\sum_{\ell=1}^{n}\I\{Y_{\ell}=1\}$
be the empirical probability of $'1'$. Assume that $p\in[\frac{2}{5},\frac{1}{2}]$,
that $n\geq60\log(\frac{4}{\delta})$, and let $\delta\in[0,\frac{1}{2}]$
be given. Then 
\begin{equation}
\left|h_{b}(\hat{p}(n))-h_{b}(p)\right|\leq\ucb_{\text{\emph{ber}}}^{(1/2)}(\hat{p}(n),\delta,n),\label{eq: entropy confidene bound close to half empirical}
\end{equation}
with probability larger than $1-\delta$, and 
\begin{equation}
\left|h_{b}(\hat{p}(n))-h_{b}(p)\right|\leq\ucb_{\text{\emph{ber}}}^{(1/2)}(p,\delta,n),\label{eq: entropy confidene bound close to half population}
\end{equation}
with probability larger than $1-\delta$. 
\end{prop}
We restricted $p_i$ to $[\frac{2}{5},\frac{1}{2}]$ since the regime of current interest is such that the arm probabilities are close to $1/2$, and this restriction simplifies the exposition. With this result, the following regret bound can be easily derived using the same
methods used in the proof of Theorem \ref{thm: UCB-Bernoulli regret}, and so its proof is omitted. For simplicity, we only state it for $K=2$ arms. 
\begin{thm} \label{thm: UCB-Bernoulli regret close to half} Assume that ${\cal X}_{i}=\{0,1\}$
and that $p_{i}\in[\frac{2}{5},\frac{1}{2}]$ for $i\in\{1,2\}$
where $\Delta=h_{b}(p_{1})-h_{b}(p_{2})$ with $p_{2}<p_{1}<\frac{1}{2}$.
Further assume that Algorithm \ref{alg:A-UCB-general} is run with
the plug-in entropy estimator $\hat{H}(\boldsymbol{Y},n)\equiv H(\hat{p}(\boldsymbol{Y},n))$
and upper confidence deviation  
\[
\ucb(\boldsymbol{Y},\delta,n)\equiv\ucb_{\text{\emph{ber}}}^{(1/2)}(\hat{p}(\boldsymbol{Y},n),\delta,n),
\]
(as defined in (\ref{eq: confidence bound Ber half})) with $\delta\equiv\delta_{\alpha}(t)=4t^{-\alpha}$
with $\alpha>2$. Then, 
\begin{equation}
R(t)\leq\frac{784\left(\frac{1}{2}-p_{2}\right)^{2}\alpha\log(t)}{\Delta}+60\alpha\log(t)+\frac{8(\alpha-1)}{\alpha-2}\cdot\Delta.\label{eq: regret bound Ber close to half}
\end{equation}
\end{thm}

In the regime above, $(\frac{1}{2}-p_2)^2=\Lambda^2=\Theta(\Delta)$ and so the regret of the algorithm is $\Theta(\log(t))$, just as the Lai-Robbins lower bound. This result can be easily extended to multiple arms, and can also be combined with the result of Theorem \ref{thm: UCB-Bernoulli regret} by taking the minimal confidence bound out of those used in Theorem \ref{thm: UCB-Bernoulli regret}
 and that of Theorem  \ref{thm: UCB-Bernoulli regret close to half}. This will result the minimum of the  regret upper bounds of both theorems. 

Finally, we consider the other extremal regime for the arms' Bernoulli probabilities, to wit, the small value regime, in which the derivative of the binary entropy function is unbounded. For concreteness, assume that $p_{1}=\gamma>0$ and $p_{2}=0$ for
a small $\gamma>0$. It then can be easily derived that the gap is $\Delta=h_{b}(p_{1})-h_{b}(p_{2})=\Theta(\gamma\log\frac{1}{\gamma})$ and that $\Dkl(p_{2}||p_{1})=\Theta(\gamma)$ (adopting the convention that $0\log(0/q)=0$ in the definition of the KL divergence, which agrees with continuity assumptions). So, the lower bound is $\Omega(\log\frac{1}{\gamma}\cdot\log t)$.
According to Theorem \ref{thm: UCB-Bernoulli regret}, the pseudo-regret
bound achieved by the algorithm is also $O(\log\frac{1}{\gamma} \cdot \log t)$, and this agrees with the lower bound. 

We may also compare the gap-independent lower bound with the gap-independent regret bound that Theorem \ref{thm: UCB-Bernoulli regret} implies. This bound is given by
\[
R(t)= O\left( \min\left\{\log t \cdot \log\left(\frac{1}{\gamma}\right),\left(\gamma\log\frac{1}{\gamma}\right)\cdot t\right\}\right)
= O\left(\log^2(t)\right),
\]
where the maximum regret is achieved for the gap $\gamma=\Theta(\frac{\log(t)}{t})$ (as can be shown by equating both terms and evaluating the order of the solution). This result \textit{roughly} matches the gap-independent lower bound of $R(t) = \Omega(\log(t))$, which can be established by modifying, e.g., the argument in \cite[Sec. 2.3 and 2.4]{slivkins2019introduction}. We next briefly describe this argument. Therein, the standard MAB problem is reduced to a best-arm identification problem (essentially, a binary hypothesis testing problem), between a uniform Bernoulli source $p_1(1)=\frac{1}{2}$ and an $\epsilon$-biased source, that is, an arm for which $p_1(1)=\frac{1}{2}-\epsilon$ for some $\epsilon>0$. Since the KL divergence is $\Theta(\epsilon^2)$, then as long as $t=\Theta(\epsilon^{-2})$, the arm identifier resulting from any MAB algorithm will have a constant fraction of errors. Consequently, since the gap in this standard MAB problem is $\Theta(\epsilon)$, a lower bound on the pseudo-regret is given by $\Theta(\epsilon t)$ which can be chosen as large as $R(t)=\Omega(\sqrt{t})$, to obtain the gap-independent bound. In the IMAB problem and the regime considered here, the KL divergence is on the order of $\Theta(\gamma)$ (instead of $\Theta(\gamma^2)$), and the gap is $\Theta(\gamma\log(\frac{1}{\gamma}))$. Repeating the same argument then leads to a lower bound of $R(t)=\Omega(\log(t))$. Intuitively, this is also not a difficult instance of the problem (just as in the regime in which the probabilities are close to half) since here it is easy to statistically distinguish between the arms with the larger entropy (the KL divergence between the distributions is linear $\Dkl(p_1(1)||p_2(1))=\Theta(\gamma)$ rather than quadratic, while the gap is only logarithmically above linear $\Theta(\gamma\log(1/\gamma))$).

\subsection{The General Alphabet Case \label{subsec:The-General-Alphabet}}
In this section, we extend the data-dependent UCD bound of the previous section to general alphabets, of cardinality larger than $2$. To this end, let $p$ and $q$ be two probability mass functions over an alphabet
${\cal Y}$. We consider the distribution-dependent functional
\begin{flalign}\label{eq:zeta_def}
\zeta(p)\dfn1-\sum_{y\in{\cal Y}}p^{2}(y),
\end{flalign}
which can be easily seen to equal $\zeta(p)=1-e^{-H^{(2)}(p)}$, 
where $H^{(2)}(p)$ is the second-order R\'{e}nyi entropy. Note that if $H^{(2)}(p)\ll1$ then
$\zeta(p)\approx H^{(2)}(p)\ll1$ too. As we shall see, $\zeta(p_{i})|{\cal X}_{i}|$ is
a measure of the effective alphabet size of the $i$th arm. In addition, $\zeta(p_{i})$ can also be accurately estimated from the data, and thus can be used by player in determining its confidence interval. The proofs of the theoretical results of this Section \ref{subsec:The-General-Alphabet} are relegated to Appendix \ref{append: subsec The-General-Alphabet}.

Let the plug-in estimator of $\zeta(p)$ be given by 
\begin{flalign}
\hat{\zeta}(n)\equiv\hat{\zeta}(\boldsymbol{Y},n)\dfn1-\sum_{y\in{\cal Y}}\hat{p}^{2}(n,y).
\end{flalign}
The proposed UCB algorithm and its regret analysis are based on the
following confidence interval function
\begin{flalign}\label{eq: confidence bound TV}
&\ucb_{\text{tv}}(\zeta,{\cal Y},\delta,n)\dfn\nonumber\\
&\hspace{1cm}3\sqrt{\frac{\zeta|{\cal Y}|}{n}}\log\left(\frac{n|{\cal Y}|}{36\zeta}\right)+\frac{3}{2}\sqrt{\frac{\log\left(\frac{2}{\delta}\right)}{n}}\log\left(\frac{n|{\cal Y}|^{2}}{9}\right)+\frac{2|{\cal Y}|^{1/2}\log^{1/4}\left(\frac{2}{\delta}\right)\log\left(n|{\cal Y}|^{2/3}\right)}{n^{3/4}},
\end{flalign}
and the following confidence interval bound for the plug-in entropy estimator:
\begin{prop}
\label{prop: TV plug-in estimator} Let $\boldsymbol{Y}=\{Y_{\ell}\}_{\ell\in[n]}$
be IID from a PMF $p$ over an alphabet ${\cal Y}$, and let $\hat{p}(n)=\{\hat{p}(n,y)\}_{y\in{\cal Y}}$
with $\hat{p}(n,y)=\frac{1}{n}\sum_{\ell=1}^{n}\I\{Y_{\ell}=y\}$
be the empirical PMF of $\boldsymbol{Y}$. Let $\delta\in[0,0.2]$
be given. Then, if $n\geq112\cdot\log\left(\frac{2}{\delta}\right)$ it holds that 
\[
\left|H(\hat{p}(n))-H(p)\right|\leq\ucb_{\text{\emph{tv}}}(\hat{\zeta}(n),{\cal Y},\delta,n),
\]
with probability larger than $1-\delta$. 
\end{prop}

To state the upper bound on the regret, we define, as before
\begin{multline}
\Gamma_{\text{tv}}(\alpha,\zeta(p_{i}),\Delta_{i},t)=\max\Bigg\{288\frac{\zeta(p_{i})}{|{\cal Y}|}\Lambda_{1}^{2}\left(\frac{2|{\cal Y}|}{3\Delta_{i}}\right),\;36230\frac{\alpha^{1/3}\log^{1/3}(t)}{|{\cal Y}|^{2/3}}\Lambda_{1}^{4/3}\left(\frac{2|{\cal Y}|}{3\Delta_{i}}\right),\\
\frac{135}{|{\cal Y}|^{2}}\Lambda_{2}\left(\frac{9|{\cal Y}|^{2}\alpha\log(t)}{\Delta_{i}^{2}}\right),\;\frac{3}{|{\cal Y}|^{2/3}}\Lambda_{4/3}\left(\frac{27|{\cal Y}|^{4/3}\alpha^{1/3}\log^{1/3}(t)}{\Delta_{i}^{4/3}}\right),30\cdot\alpha\log\left(2^{1/\alpha}t\right),119\zeta(p_{i})|{\cal Y}|\Bigg\}.\label{eq: u tv}
\end{multline}

\begin{thm}
\label{thm: UCB-TV regret}Assume that Algorithm \ref{alg:A-UCB-general}
is run with the plug-in entropy estimator
\[
\hat{H}(\boldsymbol{Y},n)\equiv H(\hat{p}(\boldsymbol{Y},n)),
\]
and upper confidence deviation  
\[
\ucb(\boldsymbol{Y},\delta,n)\equiv\ucb_{\text{\emph{tv}}}(\hat{\zeta}(\boldsymbol{Y},n),{\cal Y},\delta,n),
\]
with $\delta\equiv\delta_{\alpha}(t)=t^{-\alpha}$ with $\alpha>2$.
Then,
\begin{equation}
R(t)\leq\sum_{i\in[K]:\Delta_{i}>0}\left[\Gamma_{\text{\emph{tv}}}(\alpha,\zeta(p_{i}),\Delta_{i},t)\cdot\Delta_{i}+\frac{4(\alpha-1)}{\alpha-2}\cdot\Delta_{i}\right].\label{eq: regret bound general}
\end{equation}
\end{thm}
To show the improvement of  using $\ucb_{\text{tv}}(\zeta,{\cal Y},\delta,n)$ of (\ref{eq: confidence bound TV}) over $\ucb_{\text{bias}}(\delta,n)$ of  (\ref{eq: confidence bound bias}), we observe that for a non-asymptotic time $t$ that is upper bounded by a polynomial function of
$|\mathcal{X}_i|$,  the regret bound of Theorem \ref{thm: UCB-TV regret} scales as 
\[
\tilde{O}\left(\frac{\zeta(p_{i})|{\cal X}_{i}|+1}{\Delta_{i}}+ \frac{|{\cal X}_{i}|^{2/3}}{\Delta_{i}^{1/3}}+\zeta(p_{i})|{\cal X}_{i}|\Delta_{i}+\Delta_{i}\right),
\]
where the $\tilde{O}(\cdot)$ hides logarithmic terms in the gap,
the alphabet size and the number of rounds. For a fixed gap, the dependence
on the alphabet size is $|{\cal X}_{i}|^{2/3}\vee \zeta(p_{i})|{\cal X}_{i}|$ 
which can be much smaller than the $|{\cal X}_{i}|$ dependence obtained
in Theorem \ref{thm: UCB-Bias regret} for the biased-based UCB. For a gap-independent bound,
we only need to consider the terms which blow-up as $\Delta_{i}\downarrow0$,
and this leads to a bound of the order $\tilde{O}(\sqrt{(\zeta(p_{i})|{\cal X}_{i}|+1)t}\vee \sqrt{|{\cal X}_{i}|}\cdot t^{1/4})$, for which the leading term $\sqrt{t}$ is multiplied by roughly $\sqrt{\zeta(p_{i})|{\cal X}_{i}|}$ rather than $|{\cal X}_{i}|$ itself. 

\section{Multi-Armed Bandits with Entropy Rewards and Support Estimation}\label{sec:unknown_support}

The previous sections assume that the player knows in advance the alphabet sizes of the sampled random variables. In practice, the exact alphabet size may not be known by the player in advance. Additionally, even when the player does know in advance the alphabet size, the actual support size can be  drastically smaller than the alphabet size, and this leads to unnecessarily large confidence bounds which increase the player's regret. To address this problem, we consider the case where the alphabet size is unknown, however, it is loosely upper bounded by a parameter $\kappa$ via the assumption $\min_{x\in\mathbb{N}_+}\{p_i(x): p_i(x) > 0\} \geq \frac{1}{\kappa}$ for every arm $i\in[K]$. We next derive upper confidence bounds for unknown support and present a UCD with support size estimation that can be input to Algorithm \ref{alg:A-UCB-general}. The proof of the theoretical results derived in this section are relegated to Appendix \ref{append: proofs sec unknown_support}.

\subsection{A Concentration Inequality for Arm's Support Size}
Let $Y \sim P$ be a discrete random variable over an $\mathbb{N}_+$, such that $p(a) = P[Y = a]$ for $a \in \mathbb{N}_+$.
Assume that 
\[\min_{a\in\mathbb{N}_+}\{P(a): P(a) > 0\} \geq \frac{1}{\kappa}\]
for some given $\kappa \geq 1$, and so it holds that the support
size of $Y$ is at most $\kappa$. 
Let $\mathcal{S}(P):=\{a \in \mathbb{N}_+: P(a) > 0\}$   be the support of $P$.
Given $n$ IID samples $\boldsymbol{Y}=(Y_i)_{i\in[n]}$ from $P$, our goal is to find confidence
interval for the support size $S(P) := |\mathcal{S}(P)|$. 

For a given dataset $(Y_i)_{i\in[n]}$, let 
\[N_a(n):= \sum_{i=1}^n\mathbbm{1}[Y_i=a]\] 
be the number of appearances of $a$ in $(Y_i)_{i\in[n]}$. The simplest intuitive estimator is given by the plug-in estimator
\begin{flalign}
\hat{S}(\boldsymbol{Y},n) = \sum_{a\in\mathbb{N}_+}\mathbbm{1}[N_a(n)>0],
\end{flalign}
which is the number of different symbols that appeared in $(Y_i)_{i\in[n]}$.

\begin{prop}\label{prop:concent_ineq_support}
Let $\boldsymbol{Y}=\{Y_{\ell}\}_{\ell\in[n]}$
be IID from a discrete distribution $p$ over a \emph{finite} alphabet
${\cal Y}$ such that $p(y)\dfn\P[Y=y]$. Then, assuming $n\geq1$,
it holds for any $\delta\in(0,1)$ that 
\begin{flalign}
\hat{S}(\boldsymbol{Y},n) \leq S(P)\leq \left(\hat{S}(\boldsymbol{Y},n)+\sqrt{\frac{1}{2}\log\left(\frac{1}{\delta}\right)}\right)\cdot\left(1-e^{-\frac{n}{\kappa}}\right)^{-1},
\end{flalign}
with probability larger than $1-\delta$.
\end{prop}

\subsection{A UCB Regret Bound with Support Size Estimation}

Next, we present  a UCB policy with support size estimation and an upper bound on the resulting expected regret. 
We apply our analysis to the finite time regime where $\sqrt{\alpha\log(t)/2}<S(P)<\kappa$. In this regime the number of available samples for an arm is $o(\kappa\log \kappa)$ for which it is known that the plug-in estimator for the support size is known to grossly underestimate the true support size (see \cite{estimation_unseen_Wu_2019} for a detailed discussion).
Furthermore, the regime $\sqrt{\alpha\log(t)/2}<S(P)$ considers the case in which the bound of Proposition \ref{prop:concent_ineq_support} is non-trivial, and allows for $\delta$, chosen as $t^{-\alpha}$, to be positive.

In the results of the previous sections, both the biased-corrected entropy estimator and the confidence bounds depend on the alphabet size. Additionally, the  pseudo-regret terms depend (monotonically) on the alphabet size. So, if the true support size is much smaller than the alphabet size, and if the player is aware of this, it can significantly reduce its pseudo-regret (bound). Nonetheless, when the player does not necessarily know the support size in advance, it can  estimate it to reduce its pseudo-regret. Next, for simplicity of presentation, we focus on  the biased-corrected entropy estimation presented in Section \ref{sec: Bias corrected UCB}.  

Define
\begin{flalign}
B_{\text{SE}}(\boldsymbol{Y},\delta,n)&\dfn\log\left(1+\frac{1}{n}\left[\left(\hat{S}(\boldsymbol{Y},n)+\sqrt{\frac{1}{2}\log\left(\frac{1}{\delta}\right)}\right)\cdot\left(1-e^{-\frac{n}{\kappa}}\right)^{-1}-1\right]\right),
\end{flalign}
and 
\begin{flalign}\label{eq:ucb_uknown_support}
\ucb_{\text{SE}}(\boldsymbol{Y},\delta,n)\dfn B_{\text{SE}}(\boldsymbol{Y},\delta,n)+\sqrt{\frac{2\log^{2}(n)}{n}\log\left(\frac{2}{\delta}\right)}.
\end{flalign}

Now, we present the following upper confidence interval bound for the bias-corrected entropy estimator with an unknown support size.

\begin{prop}\label{prop:UCB_error_prob_etropy_est_unkown_support}
Let $\boldsymbol{Y}=\{Y_{\ell}\}_{\ell\in[n]}$
be IID from a discrete distribution $p$ over a \emph{finite} alphabet
${\cal Y}$ such that $p(y)\dfn\P[Y=y]$. Then, assuming $n\geq2$,
it holds for any $\delta\in(0,1)$ that
\[
\left|H(\hat{p}(n))-H(p)\right|\leq\ucb_{{\text{\emph{SE}}}}(\boldsymbol{Y},\delta,n),
\]
with probability larger than $1-2\delta$. 
\end{prop}

Using this confidence interval, we may adapt Algorithm \ref{alg:A-UCB-general} to include upper confidence bounds on the support size estimation, for each arm, and then derive an upper bound on the pseudo-regret of Algorithm \ref{alg:A-UCB-general} when the support size is estimated. To this end, denote
\begin{flalign}
&\Gamma_{\text{SE}}(\alpha,\beta,S,\kappa,\Delta,t)\dfn\max\left\{\frac{2\sqrt{\kappa}\left(\sqrt{S}+(\alpha\log(t)/2)^{\frac{1}{4}}\right)}{\{e^{\beta\Delta/2}-1\}\vee \sqrt{e^{\beta\Delta/2}-1}} ,\;15\cdot\Lambda_{2}\left(\frac{8\cdot\log(2t^{\alpha})}{(1-\beta)^{2}\Delta^{2}}\right)\right\} .\label{eq: u bias SE}
\end{flalign}

\begin{thm}\label{thm:upper_regret_uknown_support}
Let $\beta\in(0,1)$ be given,  $\delta\equiv\delta_{\alpha}(t)=t^{-\alpha}$ and $\alpha>2$. Additionally, denote by $S(P_i)$ the  support size of arm $i$ and assume that $\{\kappa_i\}_{i\in[K]}$ are given and are such that $S(P_i)\leq \kappa_i$ for all $i\in[K]$. 
Assume that Algorithm \ref{alg:A-UCB-general}
is run with a plug-in entropy estimator $\hat{H}(\boldsymbol{Y},n)\equiv H(\hat{p}(n))$, 
and upper confidence deviation $\ucb(\boldsymbol{Y},\delta,n)\equiv\ucb_{\text{\emph{SE}}}(\boldsymbol{Y},\delta,n)$.
Then, the pseudo-regret is bounded
as 
\begin{equation}\label{eq:upper_pregret_uknown_support}
R(t)\leq\sum_{i\in[K]:\Delta_{i}>0}\left[\Gamma_{\text{\emph{SE}}}(\alpha,\beta,S_i,\kappa,\Delta_i,t)\cdot\Delta_{i}+\frac{4(\alpha-1)}{\alpha-2}\cdot\Delta_{i}\right].
\end{equation}
\end{thm}

Note that $\Gamma_{\text{SE}}(\alpha,\beta,S,\kappa,\Delta,t)=O(\kappa^{2/3})$, for given $\beta$ and $\Delta$,
%\nw{What about $\Gamma$? Isn't that the relevant term? } 
whenever $S\vee\alpha\log(t)=o(\kappa^{2/3})$. In this case, estimating the support significantly decreases the pseudo-regret by an order of $\kappa^{1/3}$ with respect to that of the bias-correction entropy estimator we present in Section \ref{sec: Bias corrected UCB} where $|{\cal Y}|=\kappa$, cf. \eqref{eq: u bias}.

\section{Numerical Experiments \label{sec:Simulations}}
In this section we present numerical experiments that illustrate the average total regret achieved by Algorithm \ref{alg:A-UCB-general} for the various upper confidence bounds we develop.
We examined the setups summarized in Table \ref{table:numerical_results_setups}, each includes two arms, the subscript 1 denotes quantities of the first arm, similarly, the subscript 2 denotes quantities of the second arm.
\begin{table}[h!]
\centering
\begin{tabular}{||c |c| c| c|c|c||} 
 \hline
  &  Alphabet & PMF & Entropy [nats] & $\kappa$  values\\ [0.5ex] 
 \hline\hline
 Setup 1  & binary & \shortstack{$p_1(1)=0.25$ \\ $p_2(1)=0.01$} & \shortstack{$H_1= 0.5623$\\$H_2=0.0560$}  & $2,10,10^3,10^5$ \\ \hline
 Setup 2  & binary &  \shortstack{$p_1(1)=0.1$ \\ $p_2(1)=0.01$}& \shortstack{$H_1=0.3251$\\$H_2=0.0560$} & $2,10,10^3,10^5$\\\hline
  Setup 3  & binary &  \shortstack{$p_1(1)=0.3$ \\ $p_2(1)=0.15$}&\shortstack{$H_1=0.6109$\\$H_2=0.4227$} & $2,10,10^3,10^5$\\\hline
 Setup 4  & ternary &  \shortstack{$p_1(0)=p_1(1)=0.125$ \\ $p_2(0)=p_2(1)=0.005$} &\shortstack{$H_1=0.7356$\\$H_2=0.0629$} & $3,10,10^3,10^5$\\\hline
 Setup 5  & ternary & \shortstack{$p_1(0)=p_1(1)=0.05$ \\ $p_2(0)=p_2(1)=0.005$} & \shortstack{$H_1=0.3944$\\$H_2= 0.0629$} & $3,10,10^3,10^5$\\ \hline
  Setup 6  & ternary & \shortstack{$p_1(0)=p_1(1)=0.15$ \\ $p_2(0)=p_2(1)=0.075$} & \shortstack{$H_1= 0.8188$\\$H_2= 0.5267$} & $3,10,10^3,10^5$\\ \hline
\end{tabular}
\vspace{0.1cm}
\caption{Numerical result setups.}
\label{table:numerical_results_setups}
\end{table}

We set $\alpha=2.1$ and ran each setup for $1.5\times 10^6$ rounds. Additionally, for each setup, we averaged the total regret across $100$ Monte Carlo realizations.   

Figure \ref{fig:binary_alphabet} presents  numerical results for the binary alphabet, i.e., Setups 1-3. Additionally, Figure \ref{fig:ternary_alphabet} presents the numerical results for the ternary alphabet, i.e., Setups 4-6. 
The lines \textit{`Bias'} depict  the average total regret of  the bias-corrected confidence interval used in \eqref{eq: confidence bound bias} with the plug-in entropy estimator. To evaluate the sensitivity  of \eqref{eq: confidence bound bias} to the alphabet size compared with the support size, we examine multiple alphabet sizes $\kappa$ as is described in Table \ref{table:numerical_results_setups}.  
The lines \textit{`TV'} denote the PMF-based confidence intervals that are used with the plug-in entropy estimator. Here too  $\kappa$ denotes the alphabet size known to the player.
In the case of a binary alphabet $\kappa=2$, we take the minimum between the confidence interval \eqref{eq: confidence bound Ber} and the confidence interval \eqref{eq: confidence bound Ber half}. In the case of a larger alphabet, i.e., $\kappa\geq3$, we use the general alphabet confidence interval \eqref{eq: confidence bound TV}.
Finally, The lines \textit{`Bias SE'} depict  the average total regret of  the bias-corrected confidence interval with support estimation used in \eqref{eq:ucb_uknown_support} along with the plug-in entropy estimator.

For the case of binary support size that is known to the player, i.e., $\kappa=2$, it is evident from Figure \ref{fig:binary_alphabet} that the Bernoulli PMF-based confidence intervals  \eqref{eq: confidence bound Ber} and \eqref{eq: confidence bound Ber half} provide a significant reduction in the average total regret in comparison to the combination of the bias corrected estimator and confidence interval used in \eqref{eq: confidence bound bias}. Furthermore, we can see that as expected, the bias correction approach suffers from significantly increased regret values as the probabilities of drawing the symbol `1', i.e., $p(1)$, of the arms get closer to the  boundary points of the interval $\left[0,\frac{1}{2}\right]$. The Bernoulli PMF-based confidence intervals exhibit robustness in these regimes and have not suffered from such stark degradation in performance. It is also evident from Figure \ref{fig:binary_alphabet} that the resulted regrets of both the bias corrected and the PMF based approaches increase as the known alphabet size, i.e., $\kappa$  increases and the support size remains fixed. We note that since the alphabet size affects the biased corrected approach through a logarithmic term in \eqref{eq: confidence bound bias} it is robust to small variations in the alphabet size. In fact, even when $\kappa$ is set to $10^3$ the resulted increase in regret is small. Nonetheless, as we continue to increase $\kappa$ to $10^5$ the regret increases significantly. Interestingly, Figure \ref{fig:binary_alphabet} shows that the PMF approach is very sensitive to the choice of the alphabet size $\kappa$ where the increase in regret is noticeable even for $\kappa=10$. Finally, we can see from Figure \ref{fig:binary_alphabet} that regret of the bias corrected approach with support estimation is robust to the changes in $\kappa$ even for very large values of $\kappa$ such as $10^5$. 

As we increase the alphabet size from two (binary) to three (ternary), Figure \ref{fig:ternary_alphabet} shows that the general PMF-based confidence interval \eqref{eq: confidence bound TV} does not perform as well as the binary ones, i.e.,  \eqref{eq: confidence bound Ber} and \eqref{eq: confidence bound Ber half}. 
This occurs since the general PMF-based confidence interval  \eqref{eq: confidence bound TV} targets scenarios where   $\zeta(p_{i})\cdot|{\cal X}_{i}|$ is sufficiently small.
Furthermore, we can see that  $\kappa$ impacts the experience regret of sources with ternary support size, i.e.,  Figure \ref{fig:ternary_alphabet} in a very similar way to its impact on the regret in the case of a binary support size, see Figure \ref{fig:binary_alphabet}. 

\begin{figure*}
\centering
 \subfloat[][(a):\:\:Setup 1]{\includegraphics[scale=0.6]{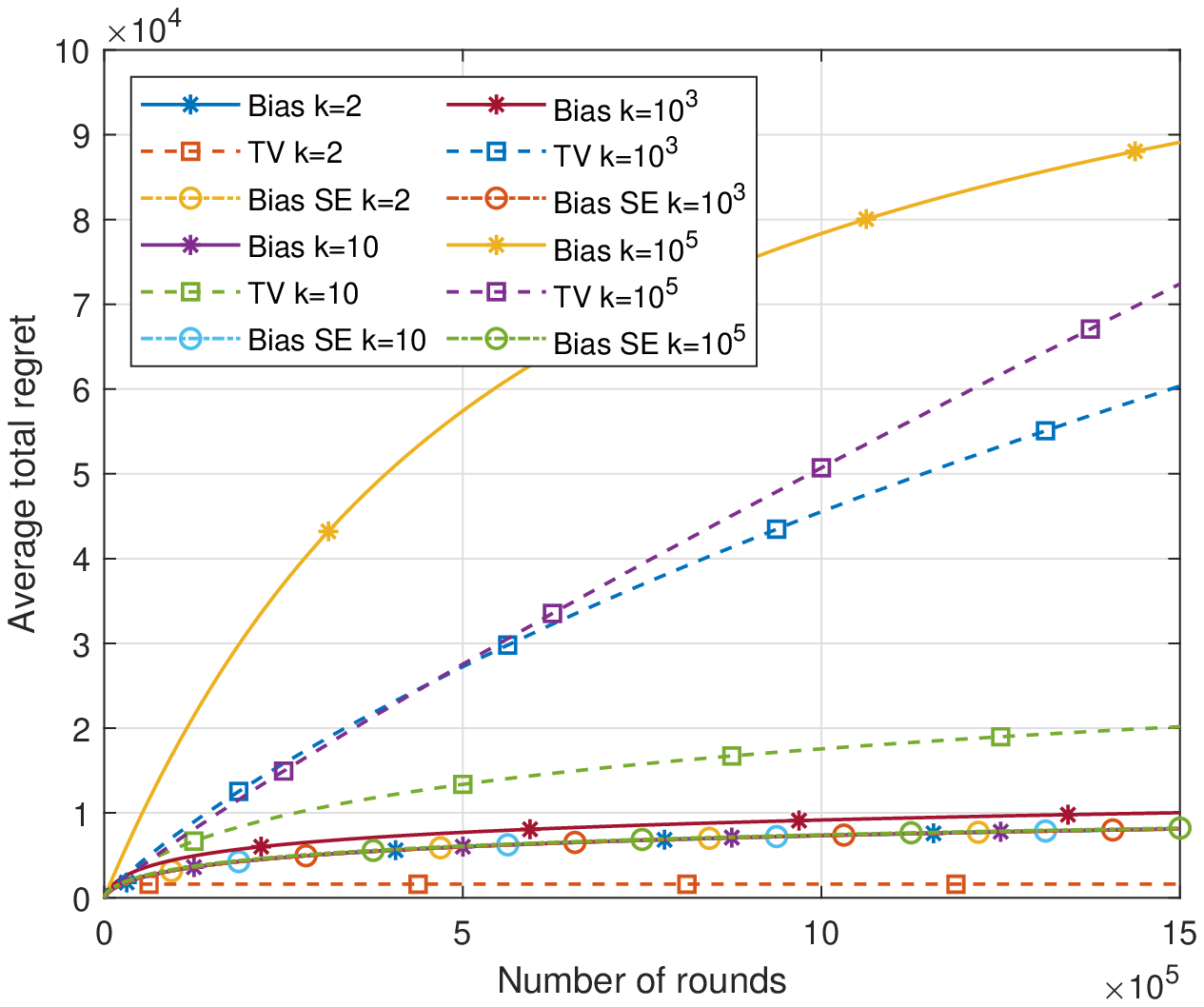}\label{fig:binary_setup1}}\hspace{0.5cm} 
 \subfloat[][(b):\:\: Setup 2]{\includegraphics[scale=0.6]{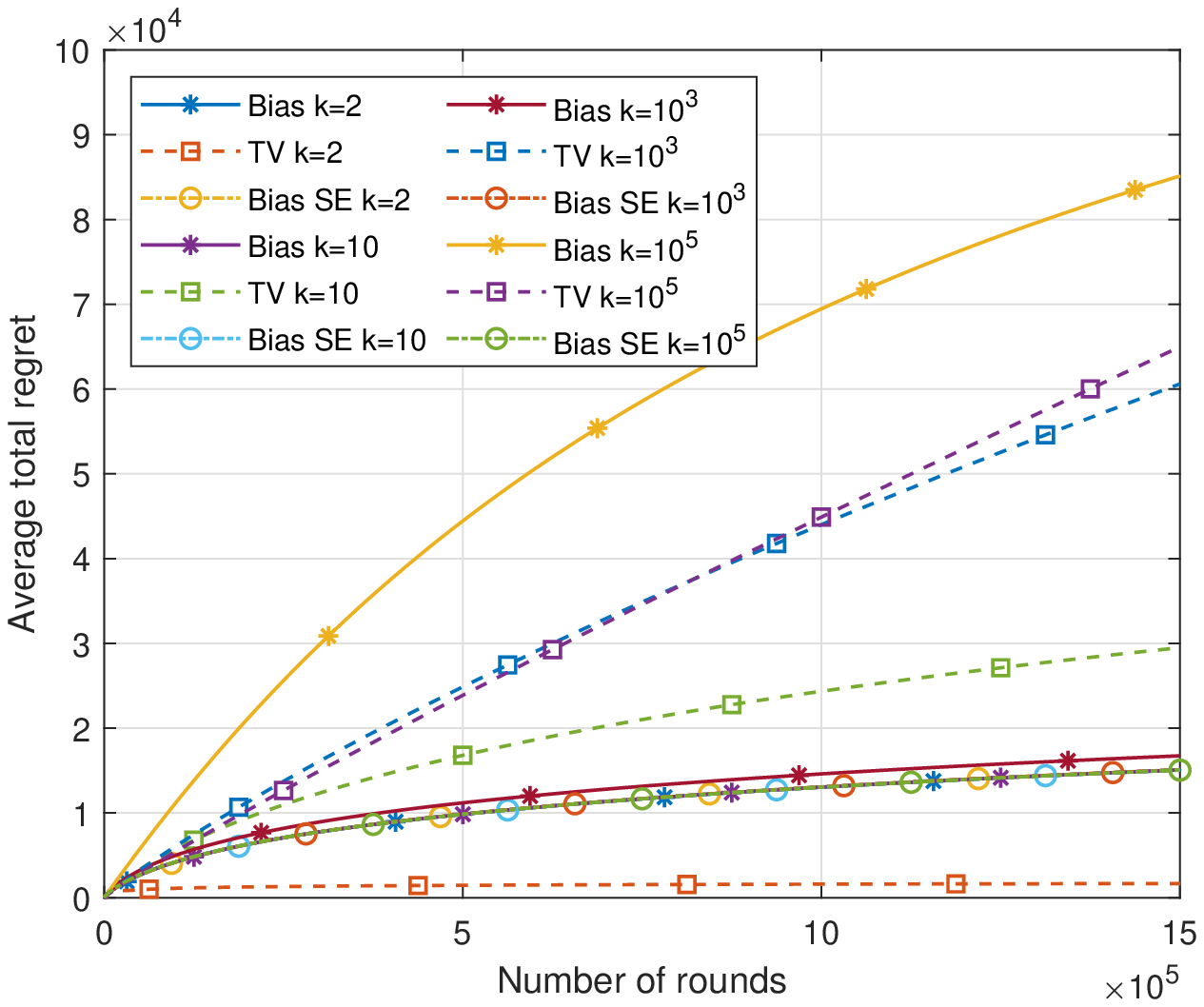}\label{fig:binary_setup2}}\\
 \subfloat[][(c):\:\:Setup 3 ]{\includegraphics[scale=0.6]{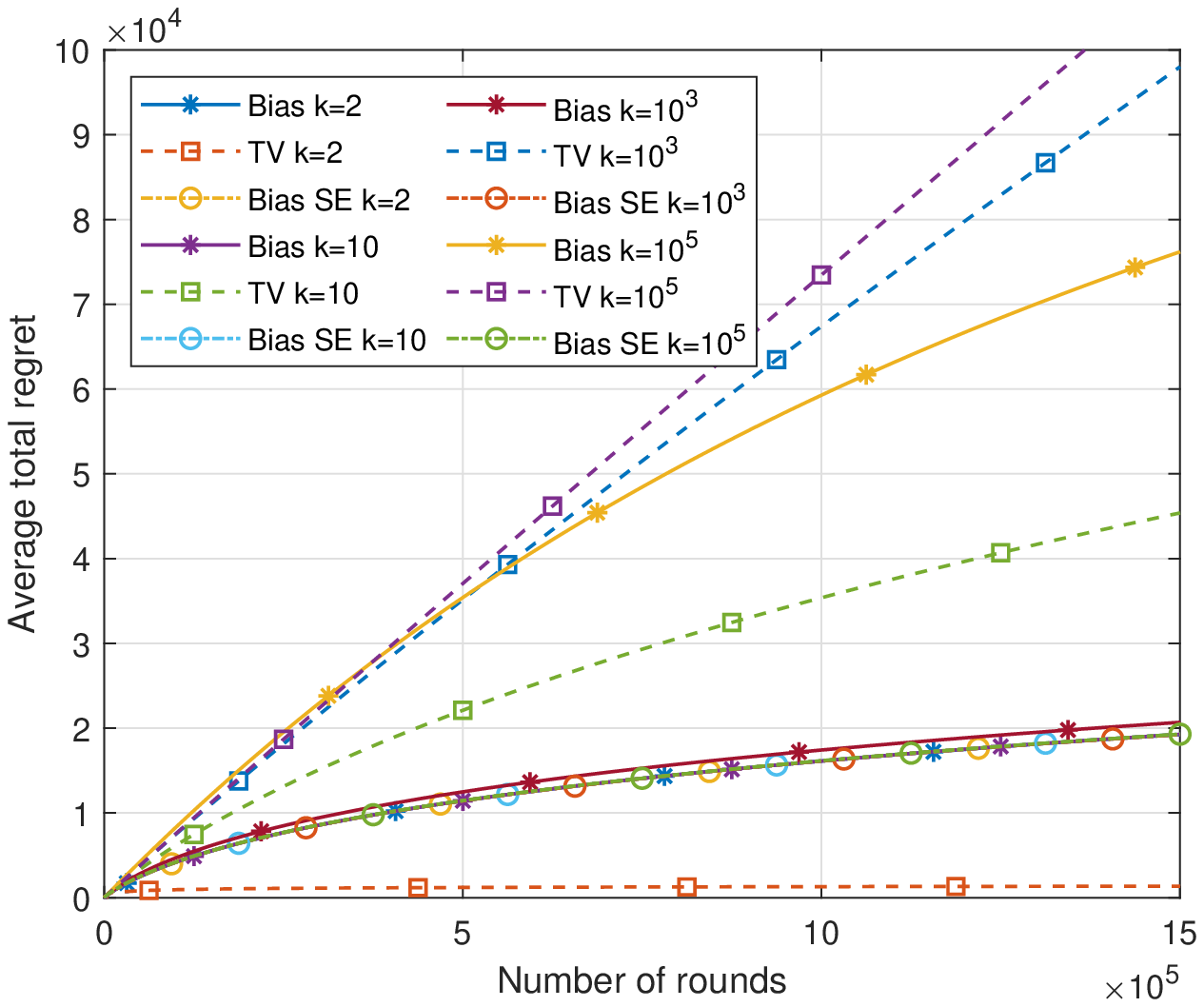}\label{fig:binary_setup3}}
 \caption{Average total regret as a function of the number of rounds for a two-armed bandit model for arms with a binary alphabet.}\label{fig:binary_alphabet}
  \vspace{-0.1in}
\end{figure*}

 \begin{figure*}
\centering
 \subfloat[][(a):\:\:Setup 4]{\includegraphics[scale=0.6]{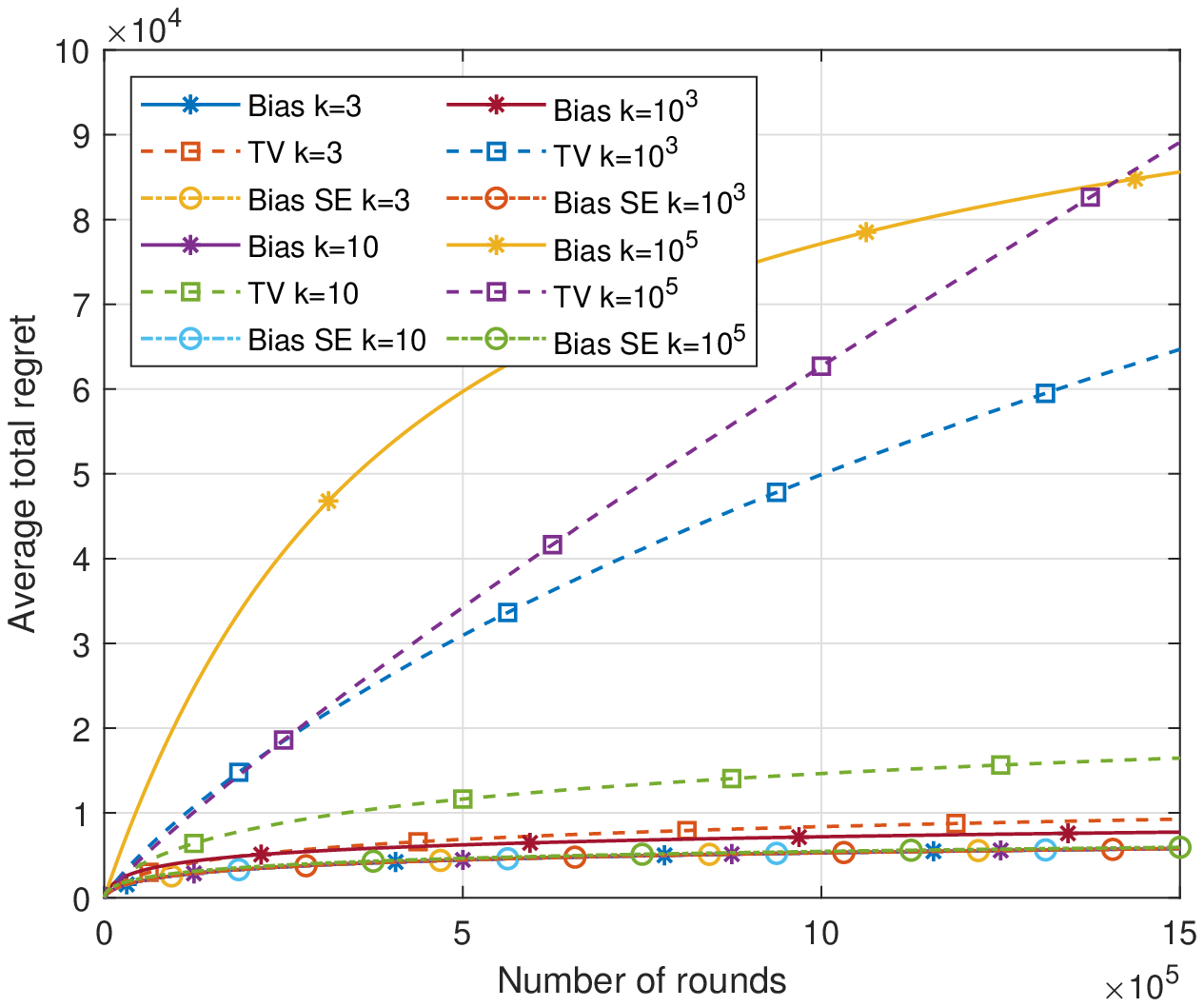}\label{fig:ternary_setup4}}\hspace{0.5cm} 
 \subfloat[][(b):\:\: Setup 5]{\includegraphics[scale=0.6]{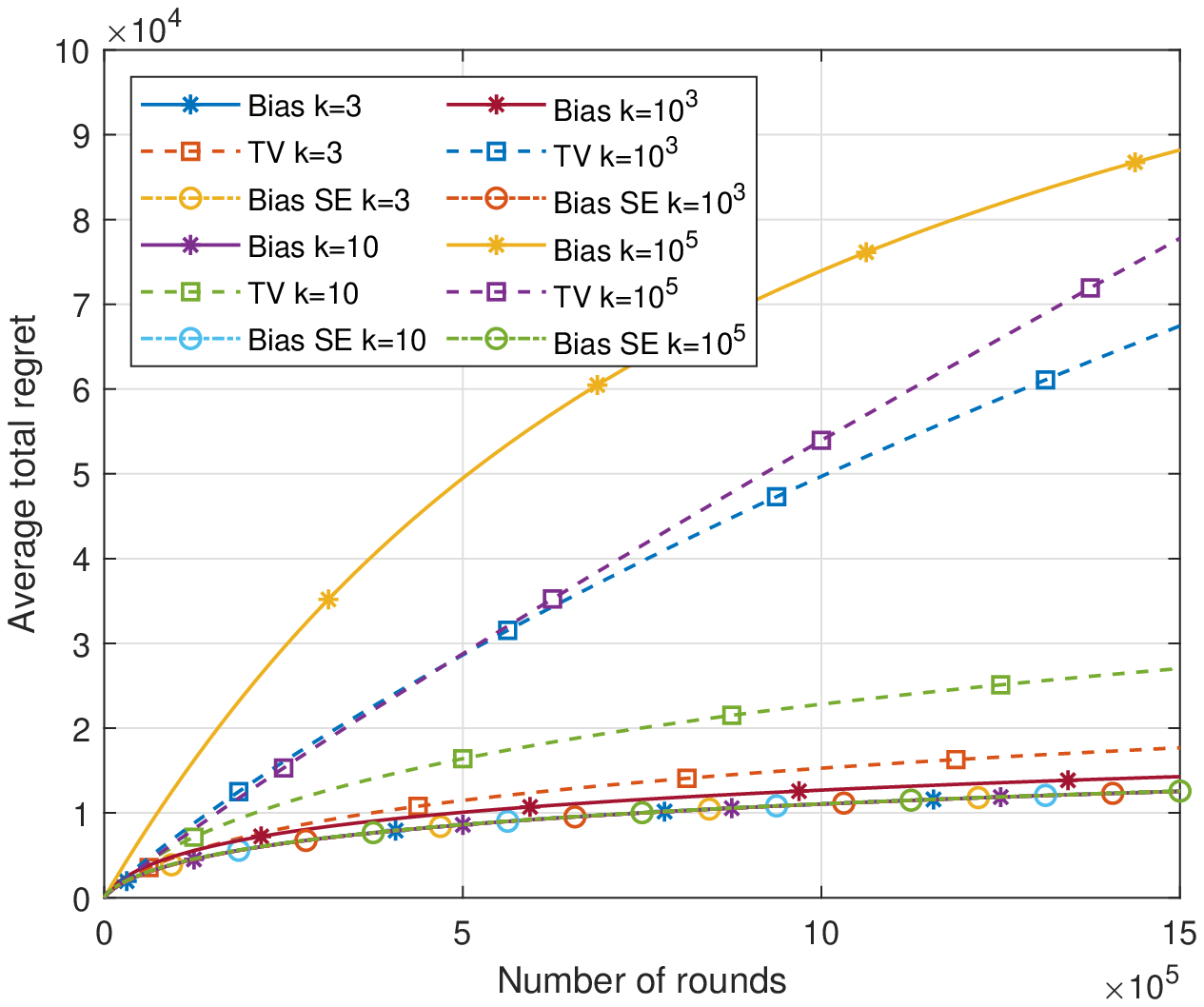}\label{fig:ternary_setup5}}\\
 \subfloat[][(c):\:\:Setup 6 ]{\includegraphics[scale=0.6]{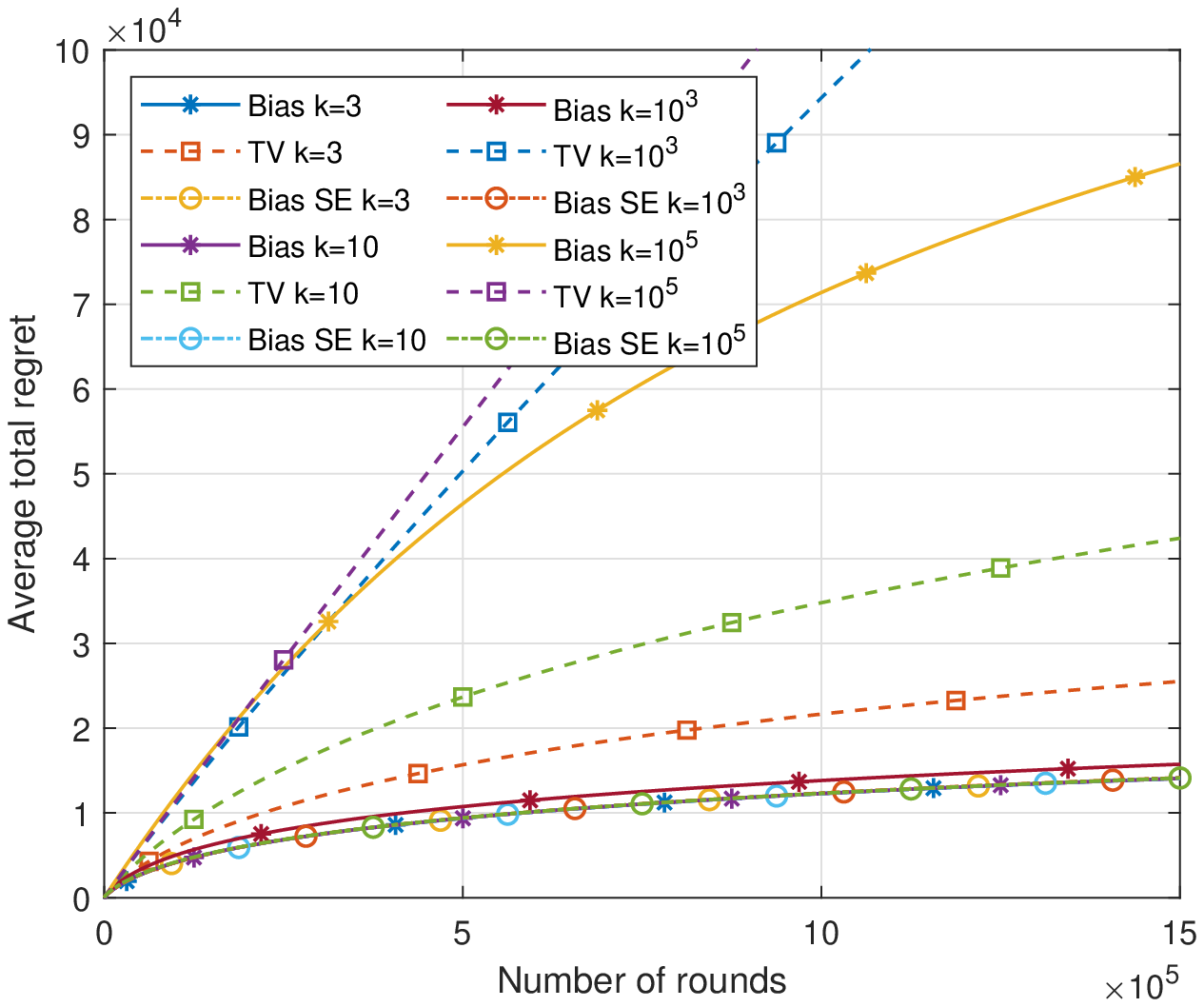}\label{fig:ternary_setup6}}
 \caption{Average total regret as a function of the number of rounds for a two-armed bandit model for arms with a ternary alphabet.}\label{fig:ternary_alphabet}
  \vspace{-0.1in}
\end{figure*}

In addition to Setups 1-6 that capture scenarios with small alphabet sizes, we consider a scenario with a large alphabet size, namely, one with $10^4$ symbols.  
For the first arm, the total probability of the first $10^4-1$ symbols is $5\times10^{-3}$, these probabilities are chosen randomly by generating $10^4-1$ IID uniform random numbers over the interval $[0,1]$ and then normalizing them to have a total probability of $5\times10^{-3}$; the probability of the last symbol is $1-5\times10^{-3}$.
For the second arm, the total probability of the first $10^4-1$ symbols is $10^{-4}$, these probabilities are chosen randomly similarly to the way they are generated in the first arm; the probability of the last symbol is $1-10^{-4}$.
Thus, it must be for all generated PFMs that $\zeta_1 \leq 0.01$ and $\zeta_2 \leq 2\times10^{-4}$.
We refer to this setup as Setup 7.  Figure \ref{fig:large_alphabet_setup_7} demonstrates the reduction in the average regret that the  general PMF-based confidence interval \eqref{eq: confidence bound TV} leads to in a non-asymptotic time regime with a large alphabet size and small total variance values. 

\begin{figure}
    \centering
    \includegraphics[scale=0.9]{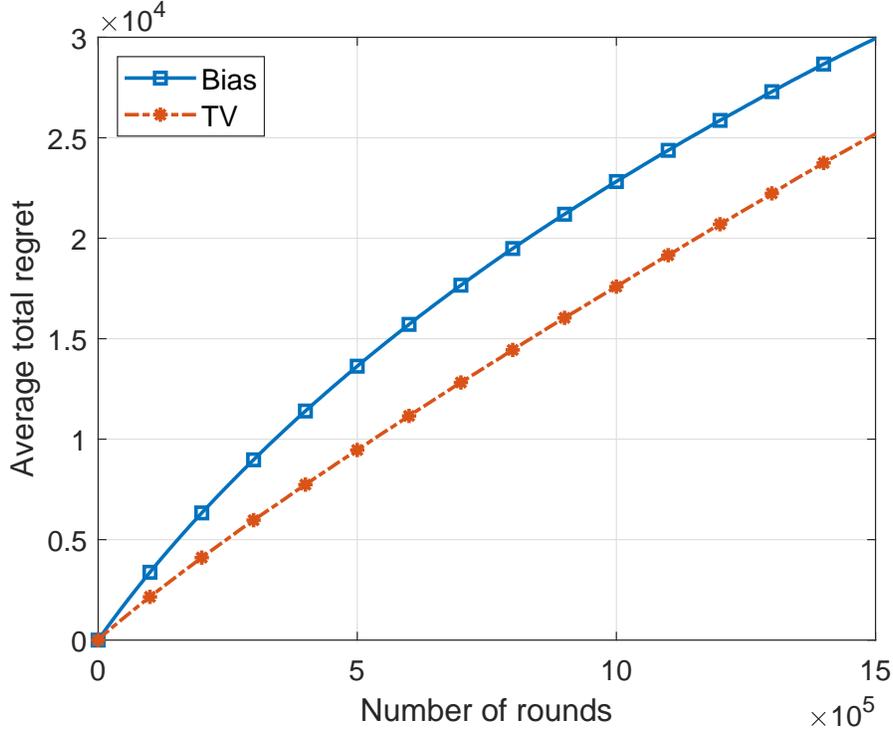}
    \caption{\centering Average total regret as a function of the number of rounds for a two-armed  \protect\linebreak bandit model for arms with an alphabet size of $10^{4}$ and a support size of $10^{4}$.\protect\linebreak\small{(Setup 7)}} 
    \label{fig:large_alphabet_setup_7}
\end{figure}

\section{Summary  \label{sec:Summary-and-Future}}
In this paper, we have introduced the IMAB problem, in which a player aims to maximize the information it observes from a set of possible sources, and concretely focused on the entropy functional. We have proposed a basic bias-corrected UCB algorithm, and showed that it is inefficient  when the entropy is very low compared to the log-alphabet size. For this regime, we have proposed a UCB algorithm that is based on data-dependent UCD, and which significantly improves upon the bias-corrected UCB algorithm. Additionally, its pseudo-regret bound agrees order-wise with Lai-Robbins impossibility lower bound in the binary alphabet case, for the gap-dependent regret bound, and almost agrees for the gap-independent regret bound. 
The first part of the paper assumes that the player knows in advance the alphabet of each arm. In practice, the alphabet of the arm may be very large compared to the support of the PMF, thus it is of interest to develop UCB algorithms for this case. To that end, we additionally developed a bias-corrected UCB algorithm with support estimation which implements a UCB approach for estimating both the support size and the resulted entropy of each arm. 

\section*{Acknowledgement}
The first author thanks Itay Ron for multiple discussions at the early stages of this research.

\appendices

\section{Inverting Polylogarithmic Functions over Linear Functions}\label{append:inverting_polylogarithmic_func}
\begin{lem}
\label{lem: inverting polylog over linear}
Let $r\in[1,2]$ be given.
There exists a constant $c_{r}>0$ so that if $x\geq c_{r}\log^{r}(1/y)/y=\Lambda_{r}(1/y)$
then $\frac{\log^{r}x}{x}\leq y$. This bound is orderwise tight as
$y\downarrow0$. Specifically, this holds for the constants $c_{1}=2$,
$c_{4/3}=3$ and $c_{2}=15$. 
\end{lem}

\begin{proof}
On $\mathbb{R}_{+}$, the function $x\to\frac{\log^{r}x}{x}$ has a
unique maximum at $x=e^{r}$, and its maximal value is $\left(\frac{r}{e}\right)^{r}$
(which is less than $1$ for any $r\in[1,2]$). So, $\frac{\log^{r}x}{x}$
is monotonic strictly decreasing for $x\geq e^{r}$. If $y\geq\left(\frac{r}{e}\right)^{r}$
then $\frac{\log x}{x}\leq y$ for all $x\in\mathbb{R}_{+}$ and the claim of the lemma trivially holds. Otherwise, if $y\in[0,\left(\frac{r}{e}\right)^{r}]$
then setting $x=c_{r}\log^{r}(1/y)/y$ results
\begin{align}
\frac{\log^{r}x}{x} & =\frac{y\log^{r}\left(\frac{c_{r}\log^{r}(1/y)}{y}\right)}{c_{r}\log^{r}(1/y)} \\
 & \leq y\cdot\left[\frac{\log(c_{r})+r\log\log(1/y)+\log(1/y)}{c_{r}^{1/r}\cdot\log(1/y)}\right]^{r}\label{eq: condition for inverse of poly-log for numerical evaluation}\\
 & \leq y\cdot\left[\frac{\log(c_{r})+(r+1)\log(1/y)}{c_{r}^{1/r}\cdot\log(1/y)}\right]^{r} \\
 & \leq y\cdot\sup_{y'\in[0,(\frac{r}{e})^{r}]}\left[\frac{\log(c_{r})+(r+1)\log(1/y')}{c_{r}^{1/r}\cdot\log(1/y')}\right]^{r}\\
 & =y\cdot\left[\frac{\log(c_{r})}{c_{r}^{1/r}r\log(\frac{e}{r})}+\frac{(r+1)}{c_{r}^{1/r}}\right]^{r}. 
\end{align}
For any given power $r$, the term inside the square brackets
can be made arbitrarily small by taking $c_{r}\uparrow\infty$, and
specifically, can be made less than $1$, which results $\frac{\log^{r}x}{x}\leq y$
for the aforementioned choice of $x$, with some numerical constant
$c_{r}$. The minimal constant can be found by checking (\ref{eq: condition for inverse of poly-log for numerical evaluation})
numerically, and this leads to the constants in the claim of the lemma.
Finally, this value of $x$ is orderwise tight since if $x=o(\log^{r}(1/y))/y$
(where the asymptotic-$o$ notation is as $y\downarrow0$), then,
$\frac{\log x}{x}/y=\omega(1)$.
\end{proof}

\section{Proofs for Section \ref{sec: Bias corrected UCB}}\label{append:proofs sec Bias corrected UCB}
\begin{proof}[Proof of Proposition \ref{prop: bias corrected plug-in estimator}]
 For the upper confidence bound it holds that 
\begin{align}
 & \P\left(H(\hat{p}(n))-H(p)>B(n)+\epsilon\right)\nonumber \\
 & =\P\left(H(\hat{p}(n))-\E(H(\hat{p}(n)))>\epsilon+\underbrace{H(p)-\E(H(\hat{p}(n)))}_{\geq0}+\underbrace{B(n)}_{\geq0}\right)\\
 & \trre[\leq,a]\P\left(H(\hat{p}(n))-\E(H(\hat{p}(n)))>\epsilon\right) \\
 & \trre[\leq,b] \exp\left[-\frac{n}{2}\left(\frac{\epsilon}{\log(n)}\right)^{2}\right],\label{eq: upper coinfidence bound on plug-in entropy-1}
\end{align}
where $(a)$ follows from the bound on the bias in  (\ref{eq:bias_negative_lower_bound}), and
$(b)$ follows from \cite[p. 168]{Antos2001}. Similarly, for the lower confidence bound
it holds that
\begin{align}
 & \P\left(H(\hat{p}(n))-H(p)<-\epsilon-B(n)\right)\nonumber \\
 & =\P\left(H(\hat{p}(n))-\E(H(\hat{p}(n)))<-\epsilon+\underbrace{H(p)-\E(H(\hat{p}(n)))-B(n)}_{\leq0}\right) \\
 & \stackrel{(a)}{\leq}\P\left(H(\hat{p}(n))-\E(H(\hat{p}(n)))<-\epsilon\right)\\
 & \stackrel{(b)}{\leq}\exp\left[-\frac{n}{2}\left(\frac{\epsilon}{\log(n)}\right)^{2}\right].\label{eq: lower coinfidence bound on plug-in entropy-1}
\end{align}
Combining (\ref{eq: upper coinfidence bound on plug-in entropy-1})
and (\ref{eq: lower coinfidence bound on plug-in entropy-1}) shows
that 
\begin{align}
\P\left(\left|H(\hat{p}(n))-H(p)\right|>B(n)+\epsilon\right) \leq  2\exp\left[-\frac{n}{2}\cdot\frac{{\epsilon}^{2}}{\log^{2}(n)}\right]\label{eq: concentration of entropy with bias}
\end{align}
for every $n\geq2$ and $\epsilon>0$. Setting the RHS of (\ref{eq: concentration of entropy with bias})
to $\delta$ and simplifying leads to the claimed result.
\end{proof}
The proof of Theorem \ref{thm: UCB-Bias regret} requires the following
lemma, which lower bounds the number of samples required for a sufficiently
low upper confidence interval. 
\begin{lem}
\label{lem: large number of samples implies UCB smaller than gap}
Let an alphabet ${\cal Y}$ be given, let a gap $\Delta\in(0,\log|{\cal Y}|]$
be given, and let $\delta=t^{-\alpha}$. Then, for any $\beta\in(0,1)$,
if $n\geq\Gamma_{\text{\emph{bias}}}(\alpha,\beta,{\cal Y},\Delta,t)$ then
$\ucb_{\text{\emph{bias}}}(t^{-\alpha},n)\leq\Delta/2$.
\end{lem}

\begin{proof}
We may assume that $n>1$. Let $\beta\in[0,1]$ be given. Then, $\ucb_{\text{bias}}(t^{-\alpha},n)\leq\Delta/2$
if both 
\begin{equation}
B(n)\leq\beta\cdot\Delta/2,\label{eq: smaller than gap first condition bias}
\end{equation}
and 
\begin{equation}
\sqrt{\frac{2\log^{2}(n)}{n}\log\left(\frac{2}{\delta}\right)}\leq(1-\beta)\cdot\Delta/2,\label{eq: smaller than gap second condition bias}
\end{equation}
holds. The first condition (\ref{eq: smaller than gap first condition bias})
is equivalent to 
\[
n\geq\frac{|{\cal Y}|-1}{e^{\beta\cdot\Delta/2}-1},
\]
and the second condition (\ref{eq: smaller than gap second condition bias})
is equivalent to 
\[
\frac{\log^{2}(n)}{n}\leq\frac{(1-\beta)^{2}\Delta^{2}}{8\log(2t^{\alpha})}.
\]
According to Lemma \ref{lem: inverting polylog over linear} this holds if 
\[
n\geq\frac{120\cdot\log(2t^{\alpha})\cdot\log^{2}\left(\frac{8\log(2t^{\alpha})}{(1-\beta)^{2}\Delta^{2}}\right)}{(1-\beta)^{2}\Delta^{2}}=15\cdot\Lambda_{2}\left(\frac{8\cdot\log(2t^{\alpha})}{(1-\beta)^{2}\Delta^{2}}\right)
\]
(recall the notation (\ref{eq: linear-times-polylog})). Simplifying
both expressions and optimizing over $\beta\in[0,1]$ concludes the proof. 
\end{proof}
With this result at hand, we may prove Theorem \ref{thm: UCB-Bias regret}. 
\begin{proof}[Proof of Thm. \ref{thm: UCB-Bias regret}]
The proof follows the analysis of \cite[Proof of Thm. 2.1]{bubeck2012regret},
with required modifications to entropy rewards structure. At round
$t$, the player chooses a suboptimal $i$ arm  with $\Delta_{i}>0$ if 
\begin{multline}
\hat{H}(\boldsymbol{X}_{i^{*}}(t-1),N_{i^{*}}(t-1)))+\ucb_{\text{bias}}(\delta_{\alpha}(t),N_{i^{*}}(t-1))\\
\leq\hat{H}(\boldsymbol{X}_{i}(t-1),N_{i}(t-1)))+\ucb_{\text{bias}}(\delta_{\alpha}(t),N_{i}(t-1)).
\end{multline}
For this to occur at least one of the following events must occur
too (sufficient conditions): \renewcommand{\labelenumi}{\Roman{enumi}.}
\renewcommand{\theenumi}{\Roman{enumi}}
\begin{enumerate}
\item \label{item:first_error} The entropy of the best arm is significantly
underestimated: 
\[
\hat{H}(\boldsymbol{X}_{i^{*}}(t-1),N_{i^{*}}(t-1)))+\ucb_{\text{bias}}(\delta_{\alpha}(t),N_{i^{*}}(t-1))\leq H_{i^{*}}.
\]
\item \label{item:second_error} The entropy of arm $i$ is significantly
overestimated:
\[
\hat{H}(\boldsymbol{X}_{i}(t-1),N_{i}(t-1)))>H_{i}+\ucb_{\text{bias}}(\delta_{\alpha}(t),N_{i}(t-1)).
\]
\item \label{item:third_error} The upper confidence interval is significantly
larger than the gap
\[
\ucb_{\text{bias}}(\delta_{\alpha}(t),N_{i}(t-1))>\Delta_{i}/2.
\]
\end{enumerate}
If all three events \ref{item:first_error}-\ref{item:third_error}
are false, then 
\begin{align}
 & \hat{H}(\boldsymbol{X}_{i^{*}}(t-1),N_{i^{*}}(t-1)))+\ucb_{\text{bias}}(\delta_{\alpha}(t),N_{i^{*}}(t-1))\nonumber \\
 & >H_{i^{*}}=H_{i}+\Delta_{i}\\
 & \geq H_{i}+2\cdot\ucb_{\text{bias}}(\delta_{\alpha}(t),N_{i}(t-1))\\
 & \geq\hat{H}(\boldsymbol{X}_{i}(t-1),N_{i}(t-1)))+\ucb_{\text{bias}}(\delta_{\alpha}(t),N_{i}(t-1)),
\end{align}
which contradicts the assumption that Algorithm \ref{alg:A-UCB-general}
chooses $I_{t}=i$ at the $t$th round. 

Next, we upper bound the expected pseudo-regret (\ref{eq: expected pseudo-regret})
of Algorithm \ref{alg:A-UCB-general} with the entropy estimator and
confidence bound stated in the theorem. To that end, we upper bound
the expected number of times a sub-optimal arm $i$ is played, i.e.,
$\E(N_{i}(t))$ as follows.  Note that if $N_{i}(t)\geq\Gamma_{\text{bias}}(\alpha,\beta,{\cal X}_{i},\Delta_{i},t)$
then event \ref{item:third_error} does not occur. so, 
\begin{align}
\E(N_{i}(t)) & =\E\left(\sum_{\tau=1}^{t}\I[I(\tau)=i]\right)\nonumber \\
 & \leq\Gamma_{\text{bias}}(\alpha,\beta,{\cal X}_{i},\Delta_{i},t)+\sum_{\tau=\Gamma_{\text{bias}}(\alpha,\beta,{\cal X}_{i},\Delta_{i},t)+1}^{t}\E\left(\I[\text{\ref{item:first_error} or \ref{item:second_error} is true a round }\tau]\right). \\
 & \leq\Gamma_{\text{bias}}(\alpha,\beta,{\cal X}_{i},\Delta_{i},t)+\sum_{\tau=1}^{t}\left[\P\left(\text{\ref{item:first_error} is true at round }\tau\right)+\P\left(\text{\ref{item:second_error} is true at round }\tau\right)\right].\label{eq: regret analysis - upper bound on the expected number of plays}
\end{align}
For any $\tau\leq t$, the first probability in (\ref{eq: regret analysis - upper bound on the expected number of plays})
is upper bounded as
\begin{align}
 & \P\left(\text{\ref{item:first_error} is true at round }\tau\right)\nonumber\\
 & \trre[\leq,a]\sum_{n=1}^{\tau}\P\left(\hat{H}(\{X_{i}(\ell)\}_{\ell\in[n]},n))+\ucb_{\text{bias}}(\delta_{\alpha}(\tau)),n)\leq H_{i}\right)\\
 & \trre[\leq,b]\tau\cdot\delta_{\alpha}(\tau)=\frac{1}{\tau^{\alpha-1}},
\end{align}
where $(a)$ follows from the union bound, and $(b)$ from the definition
of the upper confidence deviation $\ucb(\delta,n)$. The second probability
in (\ref{eq: regret analysis - upper bound on the expected number of plays})
is similarly upper bounded. Inserting these bounds back to the sum
in (\ref{eq: regret analysis - upper bound on the expected number of plays})
it then follows that 
\begin{align}
 & \sum_{\tau=1}^{t}\left[\P\left(\text{\ref{item:first_error} is true at round }\tau\right)+\P\left(\text{\ref{item:second_error} is true at round }\tau\right)\right]\nonumber \\
 & \leq2\sum_{\tau=1}^{t}\frac{1}{\tau^{\alpha-1}}\leq2\sum_{\tau=1}^{\infty}\frac{1}{\tau^{\alpha-1}} \\
 & \leq2\left[1+\int_{1}^{\infty}\frac{1}{\tau^{\alpha-1}}\d \tau\right]=\frac{2(\alpha-1)}{\alpha-2}.
\end{align}
Substituting the upper bounds in the last two displays back to (\ref{eq: regret analysis - upper bound on the expected number of plays}),
and using the resulting bound in $R(t)=\sum_{i\in[K]:\Delta_{i}>0}\E(N_{i}(t))\Delta_{i}$
then concludes the proof. \renewcommand{\labelenumi}{\arabic{enumi}.}
\renewcommand{\theenumi}{\arabic{enumi}}
\end{proof}

\section{Proofs for Section \ref{subsec:The-Bernoulli-Case}}
\label{append:proofs subsec The-Bernoulli-Case}

The proof of Proposition \ref{prop: Bernoulii plug-in estimator}
is based on a standard concentration result on the empirical mean
of a Bernoulli source.
\begin{lem}
\label{lem: empirical mean concentration Ber}In the setting of Proposition
\ref{prop: Bernoulii plug-in estimator}, each of the following events
holds with probability larger than $1-\delta$: 
\begin{equation}
\left|p-\hat{p}(n)\right|\leq\sqrt{\frac{3p\log(\frac{2}{\delta})}{n}},\label{eq: absolute deviation of empirical mean}
\end{equation}
\begin{equation}
p\leq2\hat{p}(n)+\frac{12\log(\frac{1}{\delta})}{n},\label{eq: upper scale bound on Bernoulli parameter}
\end{equation}
and 
\begin{equation}
\hat{p}(n)\leq2p+\frac{3\log(\frac{1}{\delta})}{n}.\label{eq: upper scale bound on Bernoulli parameter estimator}
\end{equation}
\end{lem}

\begin{proof}
We will use the relative (multiplicative) Chernoff bound multiple
times. This bound states that \cite[Thm. 4.4]{mitzenmacher2017probability} 
\begin{equation}
\P\left[\left|\hat{p}(n)-p\right|\geq\xi p\right]=\P\left[\hat{p}(n)-p\geq\xi p\right]+\P\left[\hat{p}(n)-p\leq-\xi p\right]\leq2e^{-\frac{\xi^{2}pn}{3}},\label{eq: relative chernoff}
\end{equation}
for any $\xi\in[0,1]$ (and it holds for the pair of one-sided deviations
each without the $2$ pre-factor). Setting $\xi=\sqrt{\frac{3\log(\frac{2}{\delta})}{pn}}$
in (\ref{eq: relative chernoff}) immediately leads to (\ref{eq: absolute deviation of empirical mean}).
Next, if $p>\frac{12\log(\frac{1}{\delta})}{n}$ then 
\[
\P\left[p\geq2\hat{p}(n)\right]=\P\left[\hat{p}(n)-p\leq-\frac{1}{2}p\right]\trre[\leq,a]e^{-\frac{pn}{12}}\trre[\leq,b]\delta,
\]
where $(a)$ is by setting $\xi=1/2$ in the one-sided version of
(\ref{eq: relative chernoff}), and $(b)$ utilizes the assumption
on $p$. Thus, with probability larger than $1-\delta$ it holds that
\[
p\leq2\hat{p}(n)\vee\frac{12\log(\frac{1}{\delta})}{n},
\]
which can be loosened to (\ref{eq: upper scale bound on Bernoulli parameter}).
Finally, If $p>\frac{3\log(\frac{1}{\delta})}{n}$ then
\[
\P\left[\hat{p}(n)>2p\right]=\P\left[\hat{p}(n)-p\geq\xi p\right]\trre[\leq,a]e^{-\frac{pn}{3}}\trre[\leq,b]\delta,
\]
where $(a)$ is by setting $\xi=1$ in the one-sided version of (\ref{eq: relative chernoff}),
and $(b)$ utilizes the assumption on $p$. Thus, with probability
larger than $1-\delta$ it holds that 
\[
\hat{p}(n)\leq2p\vee\frac{3\log(\frac{1}{\delta})}{n},
\]
which can be loosened to (\ref{eq: upper scale bound on Bernoulli parameter estimator}).
\end{proof}
The concentration of the empirical probability of the source then
leads to a confidence bound on the entropy, as next shown in the proof
of Proposition \ref{prop: Bernoulii plug-in estimator}. 
\begin{proof}[Proof of Proposition \ref{prop: Bernoulii plug-in estimator}]
If $\dtv(p,\hat{p}(n))\leq\frac{1}{2}$ then \cite[Lemma 2.7]{csiszar2011information}
implies that
\begin{align}
\left|h_{b}(\hat{p}(n))-h_{b}(p)\right| & \leq\sqrt{\frac{12p\log(\frac{2}{\delta})}{n}}\log\left(\sqrt{\frac{4n}{12p\log(\frac{2}{\delta})}}\right)\\
 & =-2\cdot\Lambda_{1}\left(\frac{\dtv(p,\hat{p}(n))}{2}\right),
\end{align}
and we note that $-\Lambda_{1}(s)$ is monotonic increasing for $s\in[0,e^{-1}]$.
For a pair of Bernoulli distributions $p$ and $q$ it holds that
\[
\dtv(p,q)=2|p(1)-q(1)|,
\]
and so by (\ref{eq: absolute deviation of empirical mean}) and (\ref{eq: upper scale bound on Bernoulli parameter})
from Lemma \ref{lem: empirical mean concentration Ber} it holds that
\begin{equation}
\dtv(p,q)\leq\sqrt{\frac{12p\log(\frac{2}{\delta})}{n}},\label{eq: high probability upper bound on TV Ber}
\end{equation}
and 
\begin{equation}
p\leq2\hat{p}(n)+\frac{12\log(\frac{1}{\delta})}{n},\label{eq: relation among Ber prob}
\end{equation}
simultaneously hold with probability larger than $1-2\delta$. To
be in the monotonic increasing regime of $-\Lambda_{1}(s)$ for any
$\hat{p}(n)$, we require that the upper bound on the total variation
distance in (\ref{eq: high probability upper bound on TV Ber}), when
substituted with the upper bound on $p$ in (\ref{eq: relation among Ber prob}),
is less than $e^{-1}$, to wit 
\[
\sqrt{\frac{12\left[2\hat{p}(n)+\frac{12\log(\frac{1}{\delta})}{n}\right]\log(\frac{2}{\delta})}{n}}\leq e^{-1}.
\]
This can be easily seen to be satisfied by the assumption $n\geq200\cdot\log(\frac{2}{\delta})$.
Now, if $2\hat{p}(n)\geq\frac{12\log(\frac{1}{\delta})}{n}$ then
$p\leq4\hat{p}(n)$ and so by the assumption of $n$ and the resulting
monotonicity,
\[
\left|h_{b}(\hat{p}(n))-h_{b}(p)\right|\leq\sqrt{\frac{12\hat{p}(n)\log(\frac{2}{\delta})}{n}}\log\left(\frac{n}{\hat{p}(n)\log(\frac{2}{\delta})}\right)
\]
(after slightly deteriorating the constants to obtain a succinct expression).
Otherwise, if $\frac{12\log(\frac{1}{\delta})}{n}\geq2\hat{p}(n)$
then $p\leq\frac{24\log(\frac{1}{\delta})}{n}$ and so by the assumption
of $n$ and the resulting monotonicity,
\[
\left|h_{b}(\hat{p}(n))-h_{b}(p)\right|\leq\frac{18\log(\frac{2}{\delta})\log(n)}{n}
\]
(after, again, slightly deteriorating the constants). To account for
both cases, we sum the two deviation terms.
Finally, to obtain (\ref{eq: confidence bound Ber}), we replace $\delta$
with $2\delta$.
\end{proof}
Next, we turn to the proof of Theorem \ref{thm: UCB-Bernoulli regret},
which is based on a lemma analogous to Lemma \ref{lem: large number of samples implies UCB smaller than gap}.
To this end, we further denote a simplified version of $\Gamma_{\text{ber}}(\cdot)$
from (\ref{eq: u ber}), defined as
\begin{equation}
\tilde{\Gamma}_{\text{ber}}(\alpha,\beta,q,\Delta,t)\dfn\max\left\{ 2\cdot\Lambda_{1}\left(\frac{36\alpha\log(t)}{(1-\beta)\Delta}\right),\frac{960q\alpha\log(t)}{\beta^{2}\cdot\Delta^{2}}\cdot\log^{2}\left(\frac{48}{\beta^{2}\cdot\Delta^{2}}\right)\right\} .\label{eq: u BER simplified}
\end{equation}

\begin{lem}
\label{lem: large number of samples implies UCB smaller than gap Bernoulli}With
$\delta\equiv\delta_{\alpha}(t)=4t^{-\alpha}$, $\:\beta\in(0,1)$ and
$\alpha>2$, if $n\geq\tilde{\Gamma}_{\text{\emph{ber}}}(\alpha,\beta,q,\Delta,t)$
then $\ucb_{\text{\emph{ber}}}(q,\delta_{\alpha}(t),n)\leq\Delta/2$ where
$\ucb_{\text{\emph{ber}}}(\cdot)$ is as defined in (\ref{eq: confidence bound Ber}). 
\end{lem}

\begin{proof}
We may assume that $n\geq e$;  this can easily be achieved by playing each arm for three rounds at the beginning of Algorithm \ref{alg:A-UCB-general}. Let $\beta\in[0,1]$ be given. Then,
$\ucb_{\text{ber}}(q,4t^{-\alpha},n)\leq\Delta/2$ if both 
\[
\sqrt{\frac{12q\alpha\log(t)}{n}}\log\left(\frac{n}{q\alpha\log(t)}\right)\leq\beta\cdot\Delta/2,
\]
and 
\[
\frac{18\alpha\log(t)\log(n)}{n}\leq(1-\beta)\cdot\Delta/2,
\]
hold. The first condition is satisfied if 
\[
\frac{\log^{2}\left(\frac{n}{q\log(\frac{6}{\delta})}\right)}{\frac{n}{q\log(\frac{6}{\delta})}}\leq\frac{\beta^{2}\cdot\Delta^{2}}{48},
\]
for which Lemma \ref{lem: inverting polylog over linear} implies
that this condition is satisfied if
\[
n\geq\frac{960q\alpha\log(t)}{\beta^{2}\cdot\Delta^{2}}\cdot\log^{2}\left(\frac{48}{\beta^{2}\cdot\Delta^{2}}\right).
\]
The second condition is satisfied if 
\[
\frac{\log(n)}{n}\leq\frac{(1-\beta)\Delta}{36\alpha\log(t)},
\]
for which Lemma \ref{lem: inverting polylog over linear} implies
that this condition is satisfied if
\[
n\geq\frac{72\alpha\log(t)}{(1-\beta)\Delta}\log\left(\frac{36\alpha\log(t)}{(1-\beta)\Delta}\right)=2\cdot\Lambda_{1}\left(\frac{36\alpha\log(t)}{(1-\beta)\Delta}\right)
\]
(recall the notation (\ref{eq: linear-times-polylog})). The claim
of the lemma then follows from the definition of $\tilde{\Gamma}_{\text{ber}}(\cdot)$
in (\ref{eq: u BER simplified}). 
\end{proof}
We may now prove Theorem \ref{thm: UCB-Bernoulli regret}. 
\begin{proof}[Proof of Theorem \ref{thm: UCB-Bernoulli regret}]
The proof is similar to the proof of Theorem \ref{thm: UCB-Bias regret},
and so we only highlight the main differences. In what follows it
will be convenient to interchangeably use both $\ucb(\boldsymbol{Y},\delta,n)$
and $\ucb_{\text{ber}}(\hat{p}(\boldsymbol{Y},n),\delta,n)$ to denote
the (same) upper confidence bound used by the algorithm. At round
$t$, the player chooses a sub-optimal $i$ arm if $\Delta_{i}>0$
and 
\begin{multline}
\hat{H}(\boldsymbol{X}_{i^{*}}(t-1),N_{i^{*}}(t-1)))+\ucb(\boldsymbol{X}_{i^{*}}(t-1),\delta_{\alpha}(t),N_{i^{*}}(t-1))\\
\leq\hat{H}(\boldsymbol{X}_{i}(t-1),N_{i}(t-1)))+\ucb(\boldsymbol{X}_{i}(t-1),\delta_{\alpha}(t),N_{i}(t-1)).
\end{multline}
For this to occur at least one of the following events must occur
too (sufficient conditions): 
\renewcommand{\labelenumi}{\Roman{enumi}'.}
\renewcommand{\theenumi}{\Roman{enumi}'}
\begin{enumerate}
\item \label{item:first_error Ber}Either the entropy of the best arm is
significantly underestimated
\[
\hat{H}(\boldsymbol{X}_{i^{*}}(t-1),N_{i^{*}}(t-1)))+\ucb(\boldsymbol{X}_{i^{*}}(t-1),\delta_{\alpha}(t),N_{i^{*}}(t-1))\leq H_{i^{*}},
\]
or 
\[
\hat{p}\left(\boldsymbol{X}_{i^{*}}(t-1),N_{i^{*}}(t-1)\right)-\frac{1}{2}p_{i^{*}}\leq-\frac{6\log(1/\delta_{\alpha}(t))}{N_{i^{*}}(t-1)}.
\]
\item \label{item:second_error Ber}Either the entropy of arm $i$ is significantly
overestimated
\[
\hat{H}(\boldsymbol{X}_{i}(t-1),N_{i}(t-1)))>H_{i}+\ucb(\boldsymbol{X}_{i}(t-1),\delta_{\alpha}(t),N_{i}(t-1)),
\]
or 
\[
\hat{p}\left(\boldsymbol{X}_{i}(t-1),N_{i}(t-1)\right)-2p_{i}\geq\frac{3\log(1/\delta_{\alpha}(t))}{N_{i}(t-1)}.
\]
\item \label{item:third_error Ber}The upper confidence interval, which
is based on an overestimation of $\hat{p}(\boldsymbol{X}_{i}(t-1),N_{i}(t-1))$
is significantly larger than the gap
\[
\ucb_{\text{ber}}\left(2p_{i}+\frac{3\log(1/\delta_{\alpha}(t))}{N_{i}(t-1)},\delta_{\alpha}(t),N_{i}(t-1)\right)>\frac{\Delta_{i}}{2},
\]
or 
\[
N_{i}(t-1)\leq200\alpha\log(t).
\]
\end{enumerate}
As in the proof of Theorem \ref{thm: UCB-Bias regret}, if all three
events \ref{item:first_error Ber}-\ref{item:third_error Ber} are
false, then 
\begin{align}
 & \hat{H}(\boldsymbol{X}_{i^{*}}(t-1),N_{i^{*}}(t-1)))+\ucb(\boldsymbol{X}_{i^{*}}(t-1),\delta_{\alpha}(t),N_{i^{*}}(t-1)) \nonumber\\
 & \geq H_{i^{*}}=H_{i}+\Delta_{i}\\
 & \geq H_{i}+2\ucb_{\text{ber}}\left(2p_{i}+\frac{3\log(1/\delta_{\alpha}(t))}{N_{i}(t-1)},\delta_{\alpha}(t),N_{i}(t-1)\right)\\
 & \trre[\geq,*]H_{i}+2\ucb\left(\boldsymbol{X}_{i}(t-1),\delta_{\alpha}(t),N_{i}(t-1)\right)\\
 & \geq\hat{H}(\boldsymbol{X}_{i}(t-1),N_{i}(t-1)))+\ucb\left(\boldsymbol{X}_{i}(t-1),\delta_{\alpha}(t),N_{i}(t-1)\right),
\end{align}
where in $(*)$ we have used the current assumption that $N_{i}(t-1)\geq200\alpha\log(t)$,
which assures that $\ucb_{\text{ber}}(q,\delta_{\alpha}(t),N_{i}(t-1))$
is a monotonically non-decreasing function of $q$. Thus, in this
case Algorithm \ref{alg:A-UCB-general} will not choose $I_{t}=i$
at the $t$th round; a contradiction. 

By Lemma \ref{lem: large number of samples implies UCB smaller than gap Bernoulli},
if
\[
N_{i}(t-1)\geq\tilde{\Gamma}_{\text{ber}}\left(\alpha,\beta,2p_{i}+\frac{3\log(1/\delta_{\alpha}(t))}{N_{i}(t-1)},\Delta_{i},t\right),
\]
then the first part of the event \ref{item:third_error Ber} does not
occur. By the definition of $\tilde{\Gamma}_{\text{ber}}(\cdot)$
in (\ref{eq: u BER simplified}), and by setting $\delta_{\alpha}(t)=4t^{-\alpha}$,
the RHS in the last equation is upper bounded as 
\begin{multline}
\max\Bigg\{2\cdot\Lambda_{1}\left(\frac{36\alpha\log(t)}{(1-\beta)\Delta_{i}}\right), \\
\frac{2560p_{i}\alpha\log(t)}{\beta^{2}\cdot\Delta_{i}^{2}}\cdot\log^{2}\left(\frac{48}{\beta^{2}\cdot\Delta_{i}^{2}}\right)+\frac{3840\alpha\log(t)}{\beta^{2}\cdot\Delta_{i}^{2}N_{i}(t-1)}\cdot\log^{2}\left(\frac{48}{\beta^{2}\cdot\Delta_{i}^{2}}\right)\Bigg\}.
\end{multline}
This can be guaranteed by requiring that $N_{i}(t-1)$ is larger than
each of the first two terms, as well as larger than twice of each
of the additive components of the third term. To conclude, a sufficient
condition for the event \ref{item:third_error Ber} not to occur is
that 
\begin{multline}
N_{i}(t-1)\geq\\
\max\left\{ 2\cdot\Lambda_{1}\left(\frac{36\alpha\log(t)}{(1-\beta)\Delta_{i}}\right),\frac{5120p_{i}\alpha\log(t)}{\beta^{2}\cdot\Delta_{i}^{2}}\cdot\log^{2}\left(\frac{48}{\beta^{2}\cdot\Delta_{i}^{2}}\right),\frac{88\sqrt{\alpha\log(t)}}{\beta\cdot\Delta_{i}}\cdot\log\left(\frac{48}{\beta^{2}\cdot\Delta_{i}^{2}}\right)\right\} \\
=\Gamma_{\text{ber}}(\alpha,\beta,p_{i},\Delta_{i},t).\label{eq: a preliminary version of u Ber}
\end{multline}
The second part of event \ref{item:third_error Ber} does not occur
if $N_{i}(t-1)\geq200\alpha\log(t)$, which is already covered by
the condition in (\ref{eq: a preliminary version of u Ber}) if we
increase the pre-constant of the second term to $6$, which is the
definition of $\Gamma_{\text{ber}}(\cdot)$ used in (\ref{eq: u ber}). 

The analysis then follows as in the proof of Theorem \ref{thm: UCB-Bias regret},
by using Lemma \ref{lem: empirical mean concentration Ber} and Proposition
\ref{prop: Bernoulii plug-in estimator} to bound the probabilities
of the events in \ref{item:first_error Ber} and \ref{item:second_error}.
Note that the condition $N_{i}(t-1)\geq200\cdot\log(\frac{4}{\delta})=200\alpha\log(t)$
required for the confidence bound to hold with high probability is
already satisfied by (\ref{eq: a preliminary version of u Ber}).
\renewcommand{\labelenumi}{\arabic{enumi}.}
\renewcommand{\theenumi}{\arabic{enumi}}
\end{proof}

\begin{proof}[Proof of Proposition \ref{prop: Bernoulii plug-in estimator close to half}]
By Taylor approximation at the point $p$, for any $q\in[0,\frac{1}{2}]$
\begin{equation}
h_{b}(q)=h_{b}(p)+h_{b}'(p)(q-p)+\frac{h_{b}''(\xi)}{2}\left(q-p\right)^{2},\label{eq: Taylor approximation for binary entropy}
\end{equation}
where $\xi\in[p,q]\cup[q,p]$. From Lemma \ref{lem: empirical mean concentration Ber},
it holds with probability larger than $1-2\delta$ that both $p\leq2\hat{p}(n)+\frac{12\log(\frac{1}{\delta})}{n}$
and $\left|p-\hat{p}(n)\right|\leq\sqrt{\frac{3p\log(\frac{2}{\delta})}{n}}$.
Under this event, since $n\geq60\log(\frac{2}{\delta})$ was assumed,
it holds that $\hat{p}(n)\geq\frac{1}{10}$. For $q\in[\frac{2}{5},\frac{1}{2}]$
it can be easily verified that
\[
\left|h_{b}'(q)\right|=\left|\log\frac{1-q}{q}\right|\leq5\left(\frac{1}{2}-q\right),
\]
and for any $q\in[\frac{1}{10},\frac{1}{2}]$ it holds that $|h_{b}''(q)|\leq12$.
Hence, by (\ref{eq: Taylor approximation for binary entropy}), and
under the high probability event
\begin{align}
\left|h_{b}(\hat{p}(n))-h_{b}(p)\right| & \leq5\left|\frac{1}{2}-p\right|\left|\hat{p}(n)-p\right|+6\left(\hat{p}(n)-p\right)^{2}\\
 & \leq7\left|\frac{1}{2}-p\right|\sqrt{\frac{\log(\frac{2}{\delta})}{n}}+\frac{9\log(\frac{2}{\delta})}{n}.
\end{align}
The proof of (\ref{eq: entropy confidene bound close to half population})
is completed by replacing $\delta$ with $2\delta$. The proof of
(\ref{eq: entropy confidene bound close to half empirical}) is similar,
with a Taylor approximation for $p$ around $\hat{p}(n)$. 
\end{proof}

\section{Proofs for Section \ref{subsec:The-General-Alphabet}}\label{append: subsec The-General-Alphabet}

The proof of Proposition \ref{prop: TV plug-in estimator} relies
on a confidence interval bound for the entropy which is based on an empirical
version of $\zeta(p)$. We begin with the following bound. 
\begin{lem}
\label{lemma:total_variation_empirical_convergence} Consider the
setting of Proposition \ref{prop: TV plug-in estimator}. Then, for
any $\delta\in(0,1)$ 
\begin{equation}
\dtv(p,\hat{p}(n))\leq\sqrt{\frac{4\zeta(p)|{\cal Y}|+\log\left(\frac{1}{\delta}\right)}{n}},\label{eq: concentration of TV general}
\end{equation}
with probability larger than $1-\delta$. 
\end{lem}

\begin{proof}
The total variation $\dtv(p,\hat{p}(n))$ satisfies a bounded difference
inequality with constant $1/n$ as a function of $(Y_{1},\ldots,Y_{n})$,
and so by McDiarmid's inequality \cite[Thm. 3.11]{van2014probability}
\[
\P\left[\left|\dtv(p,\hat{p}(n))-\E\left[\dtv(p,\hat{p}(n))\right]\right|\geq\epsilon\right]\leq e^{-2n\epsilon^{2}}.
\]
Recall that $\hat{p}(n,y)=\frac{1}{n}\sum_{\ell=1}^{n}\I\{Y_{i}=y\}$.
We next upper bound the expected value $\E[\dtv(p,\hat{p}(n))]$ as
follows: 
\begin{align}
\E\left[\dtv(p,\hat{p}(n))\right] & =\sum_{y\in{\cal Y}}\E\left[\left|p(y)-\hat{p}(n,y)\right|\right]\\
 & \leq\sum_{y\in{\cal Y}}\sqrt{\E\left[\left(p(y)-\hat{p}(n,y)\right)^{2}\right]}\\
 & =\sum_{y\in{\cal Y}}\sqrt{\frac{2}{n}p(y)(1-p(y))}\\
 & \leq|{\cal Y}|\sqrt{\frac{1}{|{\cal Y}|}\sum_{y\in{\cal X}}\frac{2}{n}p(y)(1-p(y))}\\
 & =\sqrt{\frac{2|{\cal Y}|}{n}}\sqrt{\sum_{y\in{\cal Y}}p(y)(1-p(y))}\\
 & =\sqrt{\frac{2\zeta(p)|{\cal Y}|}{n}},
\end{align}
where the two inequalities follow from Jensen's inequality. Setting
$e^{-2n\epsilon^{2}}=\delta$ directly leads to 
\[
\dtv(p,\hat{p}(n))\leq\sqrt{\frac{2\zeta(p)|{\cal Y}|}{n}}+\sqrt{\frac{1}{2n}\log\left(\frac{1}{\delta}\right)},
\]
which is further slightly loosened to the claim of the lemma using
$\sqrt{a}+\sqrt{b}\leq\sqrt{2(a+b)}$ for $a,b\in\mathbb{R}_{+}$. 
\end{proof}
Clearly, while $\zeta(p)$ controls the size confidence interval of
$\dtv(p,\hat{p}(n))$, it is a distribution-dependent quantity which
is unknown to the player, and thus required to be estimated from the
data. In this respect, the concentration of $\zeta(p)$ to its estimated
plug-in value is roughly on the same order of that of the total variation (in fact, it can be proved to be faster). Specifically, the following holds:
\begin{lem} 
\label{lem: zeta concentration is dominated by TV concentration}
Let the plug-in estimator of $\zeta(p)$ be given by $\hat{\zeta}(n)\equiv\hat{\zeta}(\boldsymbol{Y},n)\dfn1-\sum_{y\in{\cal Y}}\hat{p}^{2}(n,y)$. 
Then, under
the setting of Lemma \ref{lemma:total_variation_empirical_convergence}, for any $\delta\in(0,1)$ 
\begin{equation}
\hat{\zeta}(n)-\sqrt{\frac{18\log\left(\frac{1}{\delta}\right)}{n}}-\frac{1}{n}\leq\zeta(p)\leq\hat{\zeta}(n)+\sqrt{\frac{18\log\left(\frac{1}{\delta}\right)}{n}},\label{eq: distribution independent concentration of zeta}
\end{equation}
with probability larger than $1-\delta$.
\end{lem}

\begin{proof}
Since $|(p(\boldsymbol{Y},n)\pm\frac{1}{n})^{2}-p(\boldsymbol{Y},n)^{2}|\leq\frac{3}{n}$
for any $p(\boldsymbol{Y},n)\in[0,1]$, the plug-in estimator $\hat{\zeta}(n)\equiv\hat{\zeta}(\boldsymbol{Y},n)$
satisfies a bounded difference inequality with constant $6/n$ as
a function of $(Y_{1},\ldots,Y_{n})$, and so by McDiarmid's inequality
\cite[Thm. 3.11]{van2014probability} 
\[
\P\left[\left|\hat{\zeta}(\boldsymbol{Y},n)-\E\left[\hat{\zeta}(\boldsymbol{Y},n)\right]\right|\geq\epsilon\right]\leq e^{-\frac{n\epsilon^{2}}{18}}.
\]
The plug-in estimator $\hat{\zeta}(n)$ is biased, and easily seen
to satisfy $\E\left[\hat{\zeta}(n)\right]=\zeta(p)+\frac{\zeta(p)}{n}$.
The result follows since $\zeta(p)\in[0,1]$. 
\end{proof}
We combine Lemma \ref{lemma:total_variation_empirical_convergence}
and Lemma \ref{lem: zeta concentration is dominated by TV concentration}
to obtain a confidence interval bound which can be computed by the
player according to its empirical data. 
\begin{lem}
\label{lem: total variation empirical concentration}Under the setting
of Lemma \ref{lemma:total_variation_empirical_convergence}, any $\delta\in(0,1/e)$
it holds that 
\begin{equation}
\dtv(p,\hat{p}(n))\leq\sqrt{\frac{4\hat{\zeta}(n)|{\cal Y}|}{n}}+\sqrt{\frac{\log\left(\frac{2}{\delta}\right)}{n}}+\frac{5|{\cal Y}|^{3/4}}{n^{3/4}}\log^{1/4}\left(\frac{2}{\delta}\right),\label{eq: dtv empirical bound general}
\end{equation}
with probability larger than $1-\delta$. 
\end{lem}

\begin{proof}
By combining Lemma \ref{lemma:total_variation_empirical_convergence}
and Lemma \ref{lem: zeta concentration is dominated by TV concentration},
and a union bound, it holds with probability larger than $1-2\delta$  that
\begin{align}
\dtv(p,\hat{p}(n)) & \leq\sqrt{\frac{4\zeta(p)|{\cal Y}|+\log\left(\frac{1}{\delta}\right)}{n}}\\
 & \leq\sqrt{\frac{4\left[\hat{\zeta}(n)+\sqrt{\frac{18\log\left(\frac{1}{\delta}\right)}{n}}\right]|{\cal Y}|+\log\left(\frac{1}{\delta}\right)}{n}}\\
 & \leq\sqrt{\frac{4\hat{\zeta}(n)|{\cal Y}|}{n}}+\left[\frac{288|{\cal Y}|^{2}\log\left(\frac{1}{\delta}\right)}{n^{3}}\right]^{1/4}+\sqrt{\frac{\log\left(\frac{1}{\delta}\right)}{n}},
\end{align}
where the last inequality follows from $\sqrt{a+b}\leq\sqrt{a}+\sqrt{b}$
for $a,b\in\mathbb{R}_{+}$. The proof is completed by substituting
$\delta$ with $2\delta$.
\end{proof}
With these results at hand we may prove Proposition \ref{prop: TV plug-in estimator}. 
\begin{proof}[Proof of Proposition \ref{prop: TV plug-in estimator}]
As in the proof of Proposition \ref{prop: Bernoulii plug-in estimator}
if $\dtv(p,\hat{p}(n))\leq\frac{1}{2}$ then 
\begin{equation}
\left|H(\hat{p}(n))-H(p)\right|\leq-|{\cal Y}|\cdot\Lambda_{1}\left(\frac{\dtv(p,\hat{p}(n))}{|{\cal Y}|}\right).\label{eq: empirical bound on entropy difference proof}
\end{equation}
From Lemma \ref{lem: total variation empirical concentration}
\[
\dtv(p,\hat{p}(n))\leq\sqrt{\frac{4\hat{\zeta}(n)|{\cal Y}|}{n}}+\sqrt{\frac{\log\left(\frac{2}{\delta}\right)}{n}}+\frac{5|{\cal Y}|^{1/2}}{n^{3/4}}\log^{1/4}\left(\frac{2}{\delta}\right)\dfn a_{1}+a_{2}+a_{3},
\]
with probability larger than $1-\delta$, where $\{a_{i}\}_{i\in[3]}$
were implicitly defined. To be in the monotonic increasing regime
of $-\Lambda_{1}(s)$ of $s\in[0,e^{-1}]$, we require that this upper
bound is less than $|{\cal Y}|e^{-1}$. This can be satisfied if 
\[
\frac{\dtv(p,\hat{p}(n))}{|{\cal Y}|}\leq\sqrt{\frac{4\hat{\zeta}(n)}{n|{\cal Y}|}}+\sqrt{\frac{\log\left(\frac{2}{\delta}\right)}{n|{\cal Y}|^{2}}}+\frac{5}{n^{3/4}|{\cal Y}|^{1/2}}\log^{1/4}\left(\frac{2}{\delta}\right)\leq\frac{1}{e}.
\]
A simple sufficient condition for this can be obtained by bounding
$\hat{\zeta}(n)\leq1$ and $|{\cal Y}|\geq2$, and requiring that
each of the three terms is less than a third of $1/e$. This holds if
$n\geq112\cdot\log\left(\frac{2}{\delta}\right)$\textbf{ }and $\delta\leq0.2$.\textbf{ }

Now, since monotonicity is satisfied, we may replace $\dtv(p,\hat{p}(n))$
with the high probability upper bound (\ref{eq: dtv empirical bound general})
of Lemma \ref{lem: total variation empirical concentration}. We may
consider three cases, according to which of the terms, which we denoted by $\{a_i\}_{i=1}^3$, is the largest.
\begin{itemize}
\item If $\max_{i\in[3]}a_{i}=a_{1}$ then the upper bound (\ref{eq: dtv empirical bound general})
is less than $3a_{1}=\sqrt{\frac{36\hat{\zeta}(n)|{\cal Y}|}{n}}$.
By the monotonicity property, (\ref{eq: empirical bound on entropy difference proof})
results
\[
\left|H(\hat{p}(n))-H(p)\right|\leq\sqrt{\frac{36\hat{\zeta}(n)|{\cal Y}|}{n}}\log\left(\sqrt{\frac{n|{\cal Y}|}{36\hat{\zeta}(n)}}\right)\leq3\sqrt{\frac{\hat{\zeta}(n)|{\cal Y}|}{n}}\log\left(\frac{n|{\cal Y}|}{36\hat{\zeta}(n)}\right).
\]
\item If $\max_{i\in[3]}a_{i}=a_{2}$ then the upper bound (\ref{eq: dtv empirical bound general})
is less than $3a_{2}=\sqrt{\frac{9\log\left(\frac{2}{\delta}\right)}{n}}$.
By the monotonicity property, (\ref{eq: empirical bound on entropy difference proof})
results
\[
\left|H(\hat{p}(n))-H(p)\right|\leq\sqrt{\frac{9\log\left(\frac{2}{\delta}\right)}{n}}\log\left(\sqrt{\frac{n|{\cal Y}|^{2}}{9\log\left(\frac{2}{\delta}\right)}}\right)\leq\frac{3}{2}\sqrt{\frac{\log\left(\frac{2}{\delta}\right)}{n}}\log\left(\frac{n|{\cal Y}|^{2}}{9}\right).
\]
\item If $\max_{i\in[3]}a_{i}=a_{3}$ then the upper bound (\ref{eq: dtv empirical bound general})
is less than $3a_{3}=\frac{15|{\cal Y}|^{3/4}}{n^{3/4}}\log^{1/4}\left(\frac{2}{\delta}\right)$.
By the monotonicity property, (\ref{eq: empirical bound on entropy difference proof})
results
\begin{align}
\left|H(\hat{p}(n))-H(p)\right| & \leq\frac{15|{\cal Y}|^{1/2}\log^{1/4}\left(\frac{2}{\delta}\right)}{n^{3/4}}\log\left(\frac{n^{3/4}|{\cal Y}|^{1/2}}{15\log^{1/4}\left(\frac{2}{\delta}\right)}\right)\\
 & \leq\frac{2|{\cal Y}|^{1/2}\log^{1/4}\left(\frac{2}{\delta}\right)\log\left(n|{\cal Y}|^{2/3}\right)}{n^{3/4}}.
\end{align}
\end{itemize}
To agree with all three cases, we sum the three deviation terms, and
this completes the proof. 
\end{proof}
We next turn to the proof of Theorem \ref{thm: UCB-TV regret}, which
is based on a lemma analogous to Lemma \ref{lem: large number of samples implies UCB smaller than gap Bernoulli}.
To this end, we further denote a simplified version of $\Gamma_{\text{tv}}(\cdot)$
from (\ref{eq: u tv}) defined as
\begin{multline}
\tilde{\Gamma}_{\text{tv}}(\alpha,\zeta,\Delta,t)\dfn\\
\max\left\{ 144\frac{\zeta}{|{\cal Y}|}\Lambda_{1}^{2}\left(\frac{2|{\cal Y}|}{3\Delta}\right),\;\frac{135}{|{\cal Y}|^{2}}\Lambda_{2}\left(\frac{9|{\cal Y}|^{2}\alpha\log(t)}{\Delta^{2}}\right),\;\frac{3}{|{\cal Y}|^{2/3}}\Lambda_{4/3}\left(\frac{27|{\cal Y}|^{4/3}\alpha^{1/3}\log^{1/3}(t)}{\Delta^{4/3}}\right)\right\} .\label{eq: u TV simplified}
\end{multline}

\begin{lem}
\label{lem: large number of samples implies UCB smaller than gap general}For
$\delta\equiv\delta_{\alpha}(t)=2t^{-\alpha}$ and $\alpha>2$, if
$n\geq\tilde{\Gamma}_{\text{\emph{tv}}}(\alpha,\zeta,\Delta,t)$ then $\ucb_{\text{\emph{tv}}}(\zeta,\delta,{\cal Y},n)\leq\Delta/2$,
where $\ucb_{\text{\emph{tv}}}(\zeta,\delta,{\cal Y},n)$ is as defined in
(\ref{eq: confidence bound TV}).
\end{lem}

\begin{proof}
We may assume\footnote{This assumption can be easily achieved if the player plays each arm 3 times at the beginning of Algorithm \ref{alg:A-UCB-general}.} that $n\geq e$. Then, $\ucb_{\text{tv}}(\zeta,{\cal Y},\delta,n)\leq\Delta/2$
if all three conditions hold\footnote{Since there are three terms involved, we do not over-complicate the
analysis with additional parameter $\beta$ (see the proof of Lemma \ref{lem: large number of samples implies UCB smaller than gap Bernoulli}). }
\[
3\sqrt{\frac{\zeta|{\cal Y}|}{n}}\log\left(\frac{n|{\cal Y}|}{36\zeta}\right)\leq\Delta/6,
\]
and 
\[
\frac{3}{2}\sqrt{\frac{\log\left(\frac{2}{\delta}\right)}{n}}\log\left(\frac{n|{\cal Y}|^{2}}{9}\right)\leq\Delta/6,
\]
as well as
\[
\frac{2|{\cal Y}|^{1/2}\log^{1/4}\left(\frac{2}{\delta}\right)\log\left(n|{\cal Y}|^{2/3}\right)}{n^{3/4}}\leq\Delta/6.
\]
For the first condition, we write it equivalently as
\[
\frac{\log\left(\sqrt{\frac{n|{\cal Y}|}{36\zeta}}\right)}{\sqrt{\frac{n|{\cal Y}|}{36\zeta}}}\leq\frac{3\Delta}{2|{\cal Y}|}.
\]
Lemma \ref{lem: inverting polylog over linear} then implies that this condition is satisfied
if
\[
n\geq144\frac{\zeta}{|{\cal Y}|}\Lambda_{1}^{2}\left(\frac{2|{\cal Y}|}{3\Delta}\right).
\]
For the second condition, we write it equivalently as
\[
\frac{\log^{2}\left(\frac{n|{\cal Y}|^{2}}{9}\right)}{\frac{n|{\cal Y}|^{2}}{9}}\leq\frac{\Delta^{2}}{9|{\cal Y}|^{2}\log\left(\frac{2}{\delta}\right)}.
\]
Lemma \ref{lem: inverting polylog over linear} implies that this condition is satisfied if
\[
n\geq\frac{135}{|{\cal Y}|^{2}}\Lambda_{2}\left(\frac{9|{\cal Y}|^{2}\log\left(\frac{2}{\delta}\right)}{\Delta^{2}}\right).
\]
For the last condition, we first require a slightly stronger condition
(in terms of the numerical constant) 
\[
\frac{\log^{4/3}\left(n|{\cal Y}|^{2/3}\right)}{n|{\cal Y}|^{2/3}}\leq\frac{\Delta^{4/3}}{27|{\cal Y}|^{4/3}\log^{1/3}\left(\frac{2}{\delta}\right)}.
\]
Lemma \ref{lem: inverting polylog over linear} implies that this condition is satisfied if
\[
n\geq\frac{3}{|{\cal Y}|^{2/3}}\Lambda_{4/3}\left(\frac{27|{\cal Y}|^{4/3}\log^{1/3}\left(\frac{2}{\delta}\right)}{\Delta^{4/3}}\right).
\]
The claim of the lemma then follows from the definition of $\tilde{\Gamma}_{\text{tv}}(\cdot)$
in (\ref{eq: u TV simplified}). 
\end{proof}
We may now prove Theorem \ref{thm: UCB-TV regret}. 
\begin{proof}[Proof of Theorem \ref{thm: UCB-TV regret}]
 The proof begins as the proof of Theorem \ref{thm: UCB-Bernoulli regret}.
We then define the events: \renewcommand{\labelenumi}{\Roman{enumi}''.}
\renewcommand{\theenumi}{\Roman{enumi}''}
\begin{enumerate}
\item \label{item:first_error general} Either the entropy of the best arm
is significantly underestimated
\[
\hat{H}(\boldsymbol{X}_{i^{*}}(t-1),N_{i^{*}}(t-1)))+\ucb(\boldsymbol{X}_{i^{*}}(t-1),\delta_{\alpha}(t),N_{i^{*}}(t-1))\leq H_{i^{*}},
\]
or 
\[
\hat{\zeta}\left(\boldsymbol{X}_{i^{*}}(t-1),N_{i^{*}}(t-1)\right)-\zeta(p_{i^{*}})\leq-\sqrt{\frac{18\log\left(\frac{1}{\delta}\right)}{N_{i^{*}}(t-1)}}.
\]
\item \label{item:second_error general} Either the entropy of arm $i$
is significantly overestimated
\[
\hat{H}(\boldsymbol{X}_{i}(t-1),N_{i}(t-1)))>H_{i}+\ucb(\boldsymbol{X}_{i}(t-1),\delta_{\alpha}(t),N_{i}(t-1)),
\]
or 
\[
\hat{\zeta}\left(\boldsymbol{X}_{i}(t-1),N_{i}(t-1)\right)-\zeta(p_{i})\geq\sqrt{\frac{18\log\left(\frac{1}{\delta}\right)}{N_{i}(t-1)}}+\frac{1}{N_{i}(t-1)}.
\]
\item \label{item:third_error general} The upper confidence interval, which
is based on an overestimation of $\hat{\zeta}(\boldsymbol{X}_{i}(t-1),N_{i}(t-1))$
is significantly larger than the gap:
\[
\ucb_{\text{tv}}\left(\zeta(p_{i})+\sqrt{\frac{18\log\left(\frac{1}{\delta}\right)}{N_{i}(t-1)}},\delta_{\alpha}(t),N_{i}(t-1)\right)>\frac{\Delta_{i}}{2},
\]
or 
\begin{equation}
N_{i}(t-1)\leq\max\left\{ 30\cdot\log\left(\frac{2}{\delta}\right),119\zeta(p_{i})|{\cal Y}|\right\} .\label{eq: third condition general - number of samples}
\end{equation}
\end{enumerate}
As in the proof of Theorem \ref{thm: UCB-Bias regret}, if all three
events \ref{item:first_error general}-\ref{item:third_error general}
are false, then 
\begin{align}
 & \hat{H}(\boldsymbol{X}_{i^{*}}(t-1),N_{i^{*}}(t-1)))+\ucb(\boldsymbol{X}_{i^{*}}(t-1),\delta_{\alpha}(t),N_{i^{*}}(t-1)) \nonumber\\
 & \geq H_{i^{*}}=H_{i}+\Delta_{i}\\
 & \geq H_{i}+2\cdot\ucb_{\text{tv}}\left(\zeta(p_{i})+\sqrt{\frac{18\log\left(\frac{1}{\delta}\right)}{N_{i}(t-1)}},\delta_{\alpha}(t),N_{i}(t-1)\right)\\
 & \trre[\geq,*]H_{i}+2\cdot\ucb\left(\boldsymbol{X}_{i}(t-1),\delta_{\alpha}(t),N_{i}(t-1)\right)\\
 & \geq\hat{H}(\boldsymbol{X}_{i}(t-1),N_{i}(t-1)))+\ucb\left(\boldsymbol{X}_{i}(t-1),\delta_{\alpha}(t),N_{i}(t-1)\right),
\end{align}
where in $(*)$ we have used the current assumption that (\ref{eq: third condition general - number of samples})
does not hold, which assures that $\ucb_{\text{tv}}(\zeta,\delta_{\alpha}(t),N_{i}(t-1))$
is a monotonically non-decreasing function of $\zeta$. Thus, in this
case Algorithm \ref{alg:A-UCB-general} will not choose $I_{t}=i$
at the $t$th round; a contradiction. 

By Lemma \ref{lem: large number of samples implies UCB smaller than gap general}
if
\[
N_{i}(t-1)\geq\tilde{\Gamma}_{\text{tv}}\left(\alpha,\zeta(p_{i})+\sqrt{\frac{18\log\left(\frac{1}{\delta}\right)}{N_{i}(t-1)}},\Delta_{i},t\right),
\]
then event \ref{item:third_error general} does not occur. By the
definition of $\tilde{\Gamma}_{\text{tv}}(\cdot)$ in (\ref{eq: u TV simplified}),
and by setting $\delta_{\alpha}(t)=t^{-\alpha}$, the RHS in the last equation is upper bounded as \textbf{
\begin{multline}
\max\Bigg\{144\frac{\zeta(p_{i})}{|{\cal Y}|}\Lambda_{1}^{2}\left(\frac{2|{\cal Y}|}{3\Delta_{i}}\right)+576\frac{\sqrt{18\log\left(\frac{1}{\delta}\right)}}{|{\cal Y}|\sqrt{N_{i}(t-1)}}\Lambda_{1}^{2}\left(\frac{2|{\cal Y}|}{3\Delta_{i}}\right)\\
,\;\frac{135}{|{\cal Y}|^{2}}\Lambda_{2}\left(\frac{9|{\cal Y}|^{2}\alpha\log(t)}{\Delta_{i}^{2}}\right),\;\frac{3}{|{\cal Y}|^{2/3}}\Lambda_{4/3}\left(\frac{27|{\cal Y}|^{4/3}\alpha^{1/3}\log^{1/3}(t)}{\Delta_{i}^{4/3}}\right)\Bigg\}.
\end{multline}
}This can be guaranteed by requiring that $N_{i}(t-1)$ is larger
than twice of each of the additive components of the first term, as
well as larger than each of the second and third terms. To conclude,
a sufficient condition for the event \ref{item:third_error general}
not to occur is that \textbf{
\begin{multline}
N_{i}(t-1)\geq
\max\Bigg\{288\frac{\zeta(p_{i})}{|{\cal Y}|}\Lambda_{1}^{2}\left(\frac{2|{\cal Y}|}{3\Delta_{i}}\right),\;36230\frac{\alpha^{1/3}\log^{1/3}(t)}{|{\cal Y}|^{2/3}}\Lambda_{1}^{4/3}\left(\frac{2|{\cal Y}|}{3\Delta_{i}}\right)\\
,\;\frac{135}{|{\cal Y}|^{2}}\Lambda_{2}\left(\frac{9|{\cal Y}|^{2}\alpha\log(t)}{\Delta_{i}^{2}}\right),\;\frac{3}{|{\cal Y}|^{2/3}}\Lambda_{4/3}\left(\frac{27|{\cal Y}|^{4/3}\alpha^{1/3}\log^{1/3}(t)}{\Delta_{i}^{4/3}}\right)\Bigg\}.\label{eq: a preliminary version of u tv}
\end{multline}
}This condition, along with the second part of event \ref{item:third_error general}
is then used to define $\Gamma_{\text{ber}}(\alpha,\zeta(p_{i}),\Delta_{i},t)$
in (\ref{eq: u tv}). The analysis then follows as in the proof of
Theorem \ref{thm: UCB-Bernoulli regret}, by using Lemma \ref{lem: total variation empirical concentration}
and Proposition \ref{prop: TV plug-in estimator} to bound the probabilities
of the events in \ref{item:first_error general} and \ref{item:second_error general}.\textbf{
}\renewcommand{\labelenumi}{\arabic{enumi}.}
\renewcommand{\theenumi}{\arabic{enumi}}
\end{proof}

\section{Proofs for Section \ref{sec:unknown_support}}\label{append: proofs sec unknown_support}
The proof of Proposition \ref{prop:concent_ineq_support} relies on the following two lemmas which we utilize in upper bounding the probability $\P(|\hat{S}(\boldsymbol{Y},n)-S(P)|>\tau)$.

\begin{lem}\label{lem:bound_exp_support_est}
Let $n\geq 1$, 
\[S(P)\cdot\left[1-\exp\left(-\frac{n}{\kappa}\right)\right]\leq \E(\hat{S}(\boldsymbol{Y},n))\leq S(P).\]
\end{lem}
\begin{proof}
Without loss of generality, let us order the symbols in the support $|\mathcal{S}(P)|$ and denote them by $s_1,\ldots,s_{S(P)}$.
First, we note that
 \begin{align}
 \E(\hat{S}(\boldsymbol{Y},n)) &=\E\left(\sum_{j=1}^{S(P)}\mathbbm{1}[N_{s_j}(n)>0]>0\right)\\
 &=\sum_{j=1}^{S(P)}\left(\P\left(\mathbbm{1}[N_{s_j}(n)>0]>0\right)\right)\\
 &=\sum_{j=1}^{S(P)}\left(1-\P\left(\mathbbm{1}[N_{s_j}(n)>0]=0\right)\right).
 \end{align}
 
 Next, we lower  and upper bound the probability $\P\left(\mathbbm{1}[N_{s_j}(n)>0]=0\right)$.
 For the lower bound we use the trivial bound $\P\left(\mathbbm{1}[N_{s_j}(n)>0]=0\right)\geq0$.
 For the upper bound we use the assumption that $P(s_j)\geq\frac{1}{\kappa}$ to deduce
 \begin{flalign}
  \P\left(\mathbbm{1}[N_{s_j}(n)>0]=0\right) = (1-P(s_j))^n\leq \left(1-\frac{1}{\kappa}\right)^n\leq \exp\left(-\frac{n}{\kappa}\right),
 \end{flalign}
where the last inequality follows since $\log(1-x)\leq -x$.
\end{proof}

\begin{lem}\label{lem:concent_ineq_suppor_size}
Let $n\geq 1$, then 
\[\P\left(|\hat{S}(\boldsymbol{Y},n)-E(\hat{S}(\boldsymbol{Y},n))|\geq \epsilon\right)\leq \exp\left(-2\epsilon^2\right).\]
\end{lem}
\begin{IEEEproof}
This follows directly from  McDiarmid's inequality and since changing the outcome of one sample changes $\hat{S}(\boldsymbol{Y},n)$ by at most one.
\end{IEEEproof}
We proceed to prove Lemma \ref{prop:concent_ineq_support}.
\begin{IEEEproof}[Proof of Proposition \ref{prop:concent_ineq_support}]
We wish to bound the probability $\P(|\hat{S}(\boldsymbol{Y},n)-S(P)|>\tau)$.
By the triangle inequality
\begin{flalign}\label{eq:prob_triangle_support_dev}
 \P(|\hat{S}(\boldsymbol{Y},n)-S(P)|>\tau)\leq \P\left(|\hat{S}(\boldsymbol{Y},n)-\E(\hat{S}(\boldsymbol{Y},n))|+|\E(\hat{S}(\boldsymbol{Y},n))-S(P)|>\tau\right).
\end{flalign}
Thus, we next upper bound the gap $|\E(\hat{S}(\boldsymbol{Y},n))-S(P)|$ and the probability $\P(|\hat{S}(\boldsymbol{Y},n)-\E(\hat{S}(\boldsymbol{Y},n))|>\tilde{\tau})$. By Lemma \ref{lem:bound_exp_support_est}
\[|\E(\hat{S}(\boldsymbol{Y},n))-S(P)|\leq S(P)\exp\left(-\frac{n}{\kappa}\right).\]
Additionally, we can upper bound the probability $\P[|\hat{S}(\boldsymbol{Y},n)-E(\hat{S}(\boldsymbol{Y},n))|>\tau-S(P)\exp\left(-\frac{n}{\kappa}\right)]$
using Lemma \ref{lem:concent_ineq_suppor_size}. We then substitute $\tau=\sqrt{\frac{1}{2}\log\left(\frac{1}{\delta}\right)}+S(P)\exp\left(-\frac{n}{\kappa}\right)$ to conclude the proof.
\end{IEEEproof}

We are now ready to prove  Proposition \ref{prop:UCB_error_prob_etropy_est_unkown_support}.
\begin{IEEEproof}[Proof of Proposition \ref{prop:UCB_error_prob_etropy_est_unkown_support}]

 From Proposition 12, it holds that the event
\[
{\cal E}\dfn\left\{ \frac{\hat{S}(\boldsymbol{Y},n)+\sqrt{\frac{1}{2}\log\left(\frac{1}{\delta}\right)}}{1-e^{-\frac{n}{\kappa}}}\geq S(P)\right\} 
\]
occurs with probability $\P[{\cal E}]\geq1-\delta$. Thus, 
\begin{align}
 & \P\left[\left|H(\hat{p}(n))-H(p)\right|\geq B_{\text{SE}}(\boldsymbol{Y},\delta,n)+\sqrt{\frac{2\log^{2}(n)}{n}\log\left(\frac{2}{\delta}\right)}\right] \nonumber\\
 & =\P\left[\left|H(\hat{p}(n))-H(p)\right|\geq B_{\text{SE}}(\boldsymbol{Y},\delta,n)+\sqrt{\frac{2\log^{2}(n)}{n}\log\left(\frac{2}{\delta}\right)},{\cal E}\right] \nonumber\\
 & \hphantom{==}+\P\left[\left|H(\hat{p}(n))-H(p)\right|\geq B_{\text{SE}}(\boldsymbol{Y},\delta,n)+\sqrt{\frac{2\log^{2}(n)}{n}\log\left(\frac{2}{\delta}\right)},{\cal E}^{c}\right]\\
 & \leq\P\left[\left|H(\hat{p}(n))-H(p)\right|\geq\frac{1}{2}\log\left(1+\frac{S(P)-1}{n}\right)+\sqrt{\frac{2\log^{2}(n)}{n}\log\left(\frac{2}{\delta}\right)},{\cal E}\right]+\P\left[{\cal E}^{c}\right]\\
 & \leq\P\left[\left|H(\hat{p}(n))-H(p)\right|\geq\frac{1}{2}\log\left(1+\frac{S(P)-1}{n}\right)+\sqrt{\frac{2\log^{2}(n)}{n}\log\left(\frac{2}{\delta}\right)}\right]+\delta\\
 & \leq2\delta.
\end{align}

\end{IEEEproof}

Next, we present an upper bound on the number maximal number of rounds, i.e., $n$, such that $\ucb_{\text{SE}}(\boldsymbol{Y},t^{-\alpha},n)\leq\Delta/2$.

\begin{lem}\label{lemma:lower_bound_arm_sample_support}
Let a support ${\cal S}$ such that $|{\cal S}|=S\leq\kappa$ be given, let a gap $\Delta\in(0,\log (S)]$
be given, and let $\delta=t^{-\alpha}$. Then, for any $\beta\in(0,1)$,
if $n\geq\Gamma_{\text{\emph{SE}}}(\alpha,\beta,S,\kappa,\Delta,t)$ then
$\ucb_{\text{\emph{SE}}}(\boldsymbol{Y},t^{-\alpha},n)\leq\Delta/2$ for every $\boldsymbol{Y}\in {\cal S}^n$.
\end{lem}

\begin{IEEEproof}
By the definition of $\ucb_{\text{SE}}(\boldsymbol{Y},\delta,n)$ we have that if
\begin{flalign}
B_{\text{SE}}(\boldsymbol{Y},\delta,n)\leq \beta\cdot\frac{\Delta}{2},
\label{eq:first_cond_support_est_regret}
\end{flalign}
and
\begin{flalign}
\sqrt{\frac{2\log^{2}(n)}{n}\log\left(\frac{2}{\delta}\right)}\leq(1-\beta)\cdot\frac{\Delta}{2},
\label{eq:second_cond_support_est_regret}
\end{flalign}
then $\ucb_{\text{SE}}(\boldsymbol{Y},\delta,n)\leq\Delta_i/2$.

First, we analyze the first condition, i.e., \eqref{eq:first_cond_support_est_regret}. 
Since $\hat{S}(\boldsymbol{Y},n)\leq S$,  it is fulfilled whenever
\begin{flalign}
\log\left(1+\frac{1}{n}\left[\left(S+\sqrt{\frac{1}{2}\log\left(\frac{1}{\delta}\right)}\right)\cdot\left(1-e^{-\frac{n}{\kappa}}\right)^{-1}-1\right]\right)\leq\beta\cdot\frac{\Delta}{2}.
\end{flalign}
Further algebra yields that $B_{\text{SE}}(\boldsymbol{Y},\delta,n)\leq \beta\cdot\frac{\Delta}{2}$ for all $n$ such that
\begin{flalign}\label{eq:UCB_SE_first_cond_I_S}
 \frac{1}{e^{\beta\Delta/2}-1}\left(S+\sqrt{\frac{1}{2}\log\left(\frac{1}{\delta}\right)}\right)\leq \left(n+ \frac{1}{e^{\beta\Delta/2}-1}\right)\left(1-e^{-\frac{n}{\kappa}}\right).
\end{flalign}
Next, utilizing the bound $1-e^{-x}\geq\frac{x}{2}$ which holds for every $x\in [0,1]$, we replace the condition \eqref{eq:UCB_SE_first_cond_I_S} with the following stricter condition 
\begin{flalign}\label{eq:UCB_SE_first_cond_I_exp_bound}
 \frac{1}{e^{\beta\Delta/2}-1}\left(S+\sqrt{\frac{1}{2}\log\left(\frac{1}{\delta}\right)}\right)\leq \left(n+ \frac{1}{e^{\beta\Delta/2}-1}\right)\frac{n}{2\kappa},\quad \text{ and }\quad \frac{n}{\kappa}\leq 1.
\end{flalign}
Substituting $\delta = t^{-\alpha}$ and replacing the RHS  of \eqref{eq:UCB_SE_first_cond_I_exp_bound}, i.e. $\left(n+ \frac{1}{e^{\beta\Delta/2}-1}\right)\frac{n}{2\kappa}$, with the stricter condition $\frac{n^2}{2\kappa}$ on $n$ yields
that \eqref{eq:UCB_SE_first_cond_I_exp_bound} is fulfilled whenever
\begin{flalign}\label{eq:UCB_SE_first_cond_I_exp_bound_final}
 n\geq \sqrt{\frac{2\kappa\left(S+\sqrt{\frac{\alpha}{2}\log(t)}\right)}{e^{\beta\Delta/2}-1}},\quad \text{ and }\quad \frac{n}{\kappa}\leq 1.
\end{flalign}

Note that since the RHS of \eqref{eq:UCB_SE_first_cond_I_S} is monotonically increasing with $n$, it is sufficient to find $n_0$ that fulfills \eqref{eq:UCB_SE_first_cond_I_exp_bound_final} to conclude that every $n\geq n_0$ fulfills \eqref{eq:UCB_SE_first_cond_I_S} too.

Now, we consider the case where $n\geq\kappa$. In this scenario $1-e^{-\frac{n}{\kappa}}\geq 1-e^{-1}\geq 1/2$.
Plugging this bound, we have that \eqref{eq:UCB_SE_first_cond_I_S} is fulfilled in this case whenever $n\geq 2\cdot\frac{S+\sqrt{\alpha\log(t)/2}}{e^{\beta\Delta/2}-1}$.

Overall, we have that \eqref{eq:UCB_SE_first_cond_I_S} is fulfilled for
\begin{flalign}
    n\geq 2\sqrt{\max\left\{\kappa,\frac{S+\sqrt{\alpha\log(t)/2}}{e^{\beta\Delta/2}-1}\right\}\cdot\frac{S+\sqrt{\alpha\log(t)/2}}{e^{\beta\Delta/2}-1}}.
\end{flalign}
Since $\sqrt{\alpha\log(t)/2}<\kappa$, we can simplify the bound and conclude that \eqref{eq:UCB_SE_first_cond_I_S} is fulfilled by all $n$ such that
\begin{flalign}
n\geq\frac{2\sqrt{\kappa}\left(\sqrt{S}+(\alpha\log(t)/2)^{\frac{1}{4}}\right)}{\min\{e^{\beta\Delta/2}-1,\sqrt{e^{\beta\Delta/2}-1}\}}.
\end{flalign}

Now, we focus on the second condition, i.e., \eqref{eq:second_cond_support_est_regret}. Observe that the second condition is exactly \eqref{eq: smaller than gap second condition bias} which is investigated in the proof of Lemma \ref{lem: large number of samples implies UCB smaller than gap}. Recall the notation (\ref{eq: linear-times-polylog}), it follows that the second condition holds for all $n$ such that
\[
n\geq15\cdot\Lambda_{2}\left(\frac{8\cdot\log(2t^{\alpha})}{(1-\beta)^{2}\Delta^{2}}\right).
\]

\end{IEEEproof}

We are now ready to prove Theorem \ref{thm:upper_regret_uknown_support}.

\begin{IEEEproof}[Proof of Theorem \ref{thm:upper_regret_uknown_support}]
The structure of this proof is similar to that of Theorem \ref{thm: UCB-Bias regret}, thus we only highlight the main differences, namely, our reliance on Proposition \ref{prop:concent_ineq_support} and Lemma \ref{lemma:lower_bound_arm_sample_support} which results in the regret terms $\Gamma_{\text{SE}}(\alpha,\beta,S_i,\kappa,\Delta_i,t)$.  

At each round
$t$, the player chooses a sub-optimal $i$ arm if $\Delta_{i}>0$
and at least one of the following events occurs:
\renewcommand{\labelenumi}{\Roman{enumi}''.}
\renewcommand{\theenumi}{\Roman{enumi}''}
\begin{enumerate}
\item \label{item:first_error_sup} The entropy of the best arm is significantly
underestimated: 
\[
\hat{H}_{\text{SE}}(\boldsymbol{X}_{i^{*}}(t-1),\delta_{\alpha}(t),N_{i^{*}}(t-1)))+\ucb_{\text{SE}}(\boldsymbol{X}_{i}(t-1),\delta_{\alpha}(t),N_{i}(t-1)))\leq H_{i^{*}}.
\]
\item \label{item:second_error_sup} The entropy of arm $i$ is significantly
overestimated:
\[
\hat{H}_{\text{SE}}(\boldsymbol{X}_{i}(t-1),\delta_{\alpha}(t),N_{i}(t-1)))>H_{i}+\ucb_{\text{SE}}(\boldsymbol{X}_{i}(t-1),\delta_{\alpha}(t),N_{i}(t-1))).
\]
\item \label{item:third_error_sup} The upper confidence interval is significantly
larger than the gap
\[
\ucb_{\text{SE}}(\boldsymbol{X}_{i}(t-1),\delta_{\alpha}(t),N_{i}(t-1)))>\Delta_{i}/2.
\]
\end{enumerate}
As in the proof of Theorem \ref{thm: UCB-Bias regret}, if all three events \ref{item:first_error_sup}-\ref{item:third_error_sup}
are false, then 
\begin{align}
 & \hat{H}_{\text{SE}}(\boldsymbol{X}_{i^{*}}(t-1),N_{i^{*}}(t-1)))+\ucb_{\text{SE}}(\boldsymbol{X}_{i^{*}}(t-1),\delta_{\alpha}(t),N_{i^{*}}(t-1))\nonumber \\
 & >H_{i^{*}}=H_{i}+\Delta_{i} \\
 & \geq H_{i}+2\cdot\ucb_{\text{SE}}(\boldsymbol{X}_{i}(t-1),\delta_{\alpha}(t),N_{i}(t-1)) \\
 & \geq\hat{H}_{\text{SE}}(\boldsymbol{X}_{i}(t-1),N_{i}(t-1)))+\ucb_{\text{SE}}(\boldsymbol{X}_{i}(t-1),\delta_{\alpha}(t),N_{i}(t-1)),
\end{align}
which contradicts the assumption that Algorithm \ref{alg:A-UCB-general}
chooses $I_{t}=i$ at the $t$th round. 

By Lemma \ref{lemma:lower_bound_arm_sample_support},
if $N_i(t-1)\geq \Gamma_{\text{SE}}(\alpha,\beta,S_i,\kappa,\Delta,t)$ then
$\ucb_{\text{SE}}(\boldsymbol{X}_{i}(t-1),t^{-\alpha},N_i(t-1))\leq\Delta_i/2$ for every $\boldsymbol{X}_{i}(t-1)\in {\cal S}_i^{N_i(t-1)}$, that is, the event \ref{item:third_error_sup} does not occur.
The analysis then follows as in the proof of Theorem \ref{thm: UCB-Bias regret}. Specifically, we upper bound the probability that the events \ref{item:first_error_sup} and \ref{item:second_error_sup} occur
based on Proposition \ref{prop:UCB_error_prob_etropy_est_unkown_support} as follows:
\begin{align}
 & \sum_{\tau=1}^{t}\left[\P\left(\text{\ref{item:first_error_sup} is true at round }\tau\right)+\P\left(\text{\ref{item:second_error_sup} is true at round }\tau\right)\right]\nonumber \\
 & \leq2\sum_{\tau=1}^{t}\frac{2}{\tau^{\alpha-1}}\leq \frac{4(\alpha-1)}{\alpha-2}.
\end{align}

\end{IEEEproof}

\bibliographystyle{IEEEtran}
% Generated by IEEEtran.bst, version: 1.14 (2015/08/26)

\end{document}